\documentclass[10 pt, journal]{IEEEtran}

\makeatletter
\def\ps@headings{%
	\def\@oddhead{\mbox{}\scriptsize\rightmark \hfil \thepage}%
	
	\def\@evenhead{\scriptsize\thepage \hfil \leftmark\mbox{}}%
	
	\def\@oddfoot{}%
	
	\def\@evenfoot{}}
	
\makeatother
\usepackage[hidelinks]{hyperref}
\usepackage{cite}
\usepackage[table]{xcolor}

\usepackage{float}  
\usepackage{epsfig,bm,amsmath,amssymb,graphicx,setspace}
\usepackage[numbers,sort&compress]{natbib}
\usepackage{color,placeins}
\usepackage{multirow} 
\usepackage[linesnumbered,ruled,vlined]{algorithm2e}
\usepackage{algpseudocode}
\usepackage{colortbl}
\usepackage{color}
\usepackage{dblfloatfix}
\usepackage{turnstile}
\usepackage{multicol}
\usepackage{array}

\usepackage{enumerate}
\usepackage{turnstile}
\usepackage{subcaption}
\usepackage{arydshln,paralist}
\usepackage{soul,tabularx,booktabs}
\usepackage{csquotes}
\usepackage{multirow}
\usepackage{adjustbox}
\usepackage{float}

\usepackage{xcolor}
\usepackage{verbatim}
\usepackage[colorinlistoftodos]{todonotes} 
\usepackage{rotating}
\definecolor{usethiscolorhere}{rgb}{0.86666,0.78431,0.78431}
\usepackage{tikz}
\usetikzlibrary{mindmap}
\usetikzlibrary{calc,positioning}
\usepackage{lipsum,adjustbox}
\usetikzlibrary{decorations.pathmorphing}

\usepackage{pifont}
\newcommand{\cmark}{\ding{51}}%
\newcommand{\xmark}{\ding{55}}%

\usepackage{amsthm}
\newtheorem{theorem}{Theorem}
\newtheorem{lemma}[theorem]{Lemma}

\makeatother

\usepackage{amsmath}

\begin{document}

\title{Swarm-Net: Firmware Attestation in IoT Swarms using Graph Neural Networks and Volatile Memory}


\author{Varun Kohli$^\dagger$, Bhavya Kohli$^\dagger$, Muhammad Naveed Aman,~\IEEEmembership{Senior Member,~IEEE}, Biplab Sikdar,~\IEEEmembership{Senior Member,~IEEE}%
\thanks{V. Kohli and B. Sikdar are with the Department of Electrical and Computer Engineering, National University of Singapore, Singapore 117417. (e-mail: varun.kohli@u.nus.edu, bsikdar@nus.edu.sg).}
\thanks{B. Kohli is with the Center for Machine Intelligence and Data Science, Indian Institute of Technology Bombay, Mumbai, India 400076. (e-mail: bhavyakohli@iitb.ac.in)}
\thanks{M. N. Aman is with the School of Computing, University of Nebraska-Lincoln, Lincoln, NE 68588, USA. (e-mail: naveed.aman@unl.edu).}
\thanks{$\dagger$ marked authors have made equal contribution to this paper.}
}

 \maketitle
 
\begin{abstract}
The Internet of Things (IoT) is a network of billions of interconnected, primarily low-end embedded devices. Despite large-scale deployment, studies have highlighted critical security concerns in IoT networks, many of which stem from firmware-related issues. Furthermore, IoT swarms have become more prevalent in industries, smart homes, and agricultural applications, among others. Malicious activity on one node in a swarm can propagate to larger network sections. Although several Remote Attestation (RA) techniques have been proposed, they are limited by their latency, availability, complexity, hardware assumptions, and uncertain access to firmware copies under Intellectual Property (IP) rights. We present \emph{Swarm-Net}, a novel swarm attestation technique that exploits the inherent, interconnected, graph-like structure of IoT networks along with the runtime information stored in the Static Random Access Memory (SRAM) using Graph Neural Networks (GNN) to detect malicious firmware and its downstream effects. We also present the first datasets on SRAM-based swarm attestation encompassing different types of firmware and edge relationships. In addition, a secure swarm attestation protocol is presented. \emph{Swarm-Net} is not only computationally lightweight but also does not require a copy of the firmware. It achieves a $99.96\%$ attestation rate on authentic firmware, $100\%$ detection rate on anomalous firmware, and $99\%$ detection rate on propagated anomalies, at a communication overhead and inference latency of $\sim1$ second and $\sim10^{-5}$ seconds (on a laptop CPU), respectively. In addition to the collected datasets, \emph{Swarm-Net}'s effectiveness is evaluated on simulated trace replay, random trace perturbation, and dropped attestation responses, showing robustness against such threats. Lastly, we compare \emph{Swarm-Net} with past works and present a security analysis.

\end{abstract}



\section{Introduction}
\label{sec:introduction}

The Internet of Things (IoT) has emerged as a leading force behind smart city initiatives for agriculture, healthcare, homes, industry, transportation security, and supply chains, owing to advances in 5G, artificial intelligence, and edge computing technologies \cite{chamola2020comprehensive}. Despite billions of dollars worth of investment worldwide \cite{forbes}, its constituent devices, which are mostly low-end embedded systems with limited computational capabilities, are often the target of various security threats \cite{hassija2019survey}. 

A recent study highlighted that 95\% of security vulnerabilities in IoT networks stem from firmware-related problems \cite{ilascu2019their}. Thus, evaluating firmware integrity is an essential part of ensuring security and trust in IoT networks. Remote Attestation (RA) has emerged as an important field of research for this purpose, and various single-node software \cite{ankergaard2021state}, hardware \cite{sfyrakis2020survey}, or hybrid \cite{johnson2021taxonomy} RA techniques have been proposed. However, their applicability to larger networks is often naive, and therefore, efficient swarm attestation approaches have also been proposed \cite{ambrosin2020collective}. Among available attestation techniques, several software-based methods rely on computationally intensive checksums over the IoT device's flash memory and offer low availability. They also require a copy of the firmware, which may not be possible under the manufacturers' Intellectual Property (IP) rights \cite{seshadri2004swatt,seshadri2006scuba,seshadri2008sake}. Further, hardware and hybrid RA techniques assume the availability of Trusted Platform Modules (TPM), Memory Protection Units (MPU), and Trusted Execution Environments (TEE), which do not apply to low-end devices \cite{agrawal2015program,brasser2015tytan,kibret2023property}. Lastly, studies based on network-flow information have limited detection capabilities against malicious firmware with normal network activity \cite{protogerou2021graph}. 

A recent paper \cite{aman2022machine} made minimal hardware assumptions on IoT devices and used Static Random Access Memory (SRAM) as a feature for machine learning-based attestation. However, their method targeted single-node RA, used the entire SRAM, and was limited to a few anomalous firmware classes in a classification-based approach. We show that the SRAM data section (a smaller part of the SRAM) is useful for detecting nodes infected by malicious firmware and propagating anomalies when combined with Graph Neural Networks (GNNs) for anomaly detection. In addition, the SRAM is significantly smaller than flash memory and, thus, easier to traverse. It captures runtime information, can indicate roving malware, and eliminates the need for firmware copies during attestation. Furthermore, 

\begin{table*}[t]
\centering
\caption{Qualitative comparison of related works on single-node and swarm attestation.}
\label{tab:related}
\resizebox{\textwidth}{!}{
\begin{tabular}{|l|c|c|c|c|c|c|c|c|c|c|c|c|}
\hline
\textbf{Evaluation Criteria}             & \textbf{\cite{seshadri2004swatt,seshadri2006scuba,seshadri2008sake}} & \textbf{\cite{aman2022machine}} & \textbf{\cite{krauss2007detecting,agrawal2015program,tan2011tpm}} & \textbf{\cite{eldefrawy2012smart,koeberl2014trustlite,brasser2015tytan}} & \textbf{HAtt \cite{aman2020hatt}} & \textbf{\cite{khodari2019decentralized,asokan2015seda,carpent2017lightweight,visintin2019safe,kuang2019esdra,dushku2020sara,ibrahim2019healed}} & \textbf{WISE \cite{ammar2018wise}} & \textbf{FeSA \cite{kuang2022fesa}} & \textbf{\cite{protogerou2021graph}} & \textbf{RAGE \cite{chilese2024one}} & \textbf{\cite{kibret2023property}} & \textbf{Swarm-Net}    \\ \hline\hline

Target                        & Single                                                                                & Single                                               & Single                                                                              & Single                                                                                    & Single                                        & Swarm                                                                                                                                                           & Swarm                                         & Swarm                                                  & Swarm                                                  & Swarm                                                  & Swarm                                                  & Swarm                   \\ \hline

RA methodology                        & S/W                                                                                & S/W                                               & H/W                                                                              & H/b                                                                                    & H/b                                        & H/b                                                                                                                                                           & H/b                                         & S/W                                                  & S/W                                                  & H/b                                                  & S/W                                                  & S/W                   \\ \hline

Approach                        & Crypto                                                                                & ML                                               & Crypto                                                                              & Crypto                                                                                    & Crypto                                        & Crypto                                                                                                                                                           & Crytpo + ML                                         & FL                                                  & GNN                                                  & GNN                                                  & ML                                                  & GNN                   \\ \hline

Feature used                      & Flash                                                                                 & SRAM                                             & Flash                                                                        & Flash                                                                                     & Flash                                         & Flash                                                                                                                                                            & Flash                                               & Property                                          & Network                                         & Execution trace                                      & Property                                          & SRAM                  \\ \hline
Low availability & \cmark                                                                 & \xmark                            & \xmark                                                               & \xmark                                                                     & \xmark                         & \xmark                                                                                                                                            & \xmark                               & \xmark                               & \xmark                                & \xmark                                & \xmark                               & \xmark \\ \hline
High latency       & \cmark                                                                 & \cmark                            & \xmark                                                               & \xmark                                                                     & \xmark                         & \xmark                                                                                                                                            & \xmark                               & \xmark                               & \xmark                                & \xmark                                & \xmark                               & \xmark \\ \hline
Interrupts disabled           & \cmark                                                                 & \xmark                            & \xmark                                                               & \xmark                                                                     & \xmark                         & \xmark                                                                                                                                            & \xmark                               & \xmark                               & \xmark                                & \xmark                                & \xmark                               & \xmark \\ \hline
Specific hardware             & \xmark                                                                 & \xmark                            & \cmark                                                               & \cmark                                                                     & \xmark                         & \cmark                                                                                                                                            & \cmark                               & \xmark                               & \xmark                                & \cmark                                & \xmark                               & \xmark \\ \hline
Homogenous devices            & \xmark                                                                 & \xmark                            & \cmark                                                               & \cmark                                                                     & \xmark                         & \cmark                                                                                                                                            & \xmark                               & \xmark                               & \xmark                                & \xmark                                & \xmark                               & \xmark \\ \hline
Firmware needed               & \cmark                                                                 & \xmark                            & \cmark                                                               & \cmark                                                                     & \cmark                         & \cmark                                                                                                                                            & \cmark                               & \xmark                               & \xmark                                & \xmark                                & \xmark                               & \xmark \\ \hline
Limited attacks               & \cmark                                                                 & \cmark                            & \cmark                                                               & \cmark                                                                     & \cmark                         & \cmark                                                                                                                                            & \cmark                               & \cmark                               & \cmark                                & \xmark                                & \cmark                               & \xmark \\ \hline
\end{tabular}}
\end{table*}

\textit{\textbf{Contributions:}} To the best of our knowledge, this paper is the first to use SRAM as a feature for swarm attestation and presents the first datasets on SRAM-based swarm attestation. The contributions of this study are listed below.
\begin{enumerate}
    \item A novel GNN-based lightweight method for firmware attestation in swarms. The proposed method uses the SRAM's data section and simple GNN designs.
    \item A secure swarm attestation protocol and the corresponding security analysis. 
    \item Datasets on SRAM-based swarm attestation for (physically deployed) four four-node and six-node swarm configurations, encompassing various complex node-level and inter-node behaviors.
    \item Thorough experimental results of different GNN architectures using the collected datasets and additional simulated scenarios. 
    \item A comparison with related works on remote and swarm attestation.
\end{enumerate}

The codes for firmware, dataset collection, and attestation are available on GitHub\footnote{\url{https://github.com/VarunKohli18/swarm-net}}. The remainder of this paper is organized as follows: Section \ref{sec:related} discusses related works on attestation and presents a qualitative comparison of the proposed method with related works. Section \ref{sec:background} discusses background concepts on SRAM and GNN, followed by Section \ref{sec:network}, which presents the network and threat models considered in this study. Section \ref{sec:proposed} introduces, \emph{Swarm-Net}, the proposed swarm attestation method, followed by a detailed discussion on the dataset, devices, hyperparameters, and evaluation metrics in Section \ref{sec:setup}. Section \ref{sec:results} shows the experimental results, and Section \ref{sec:security} analyses the security of the proposed method. The paper is concluded in Section \ref{sec:conclusion}.

\section{Related Works}
\label{sec:related}

Several techniques for single-node RA and swarm attestation have been proposed to solve firmware-related problems in IoT. We structure our discussion into three parts: (a) on specific software (S/W), hardware (H/W), and hybrid (H/b) RA techniques for single prover scenarios, followed by (b) swarm attestation approaches, and lastly, (c) a brief overview of \emph{Swarm-Net}. Table \ref{tab:related} presents a qualitative comparison of \emph{Swarm-Net} with related works.

\subsection{RA methodologies}

Traditional software-based RA techniques in single-prover scenarios typically require an IoT microcontroller to evaluate a checksum on its program memory, which the verifier validates using a copy of the hash digest \cite{seshadri2004swatt,seshadri2006scuba,seshadri2008sake}. SWATT \cite{seshadri2004swatt}, for instance, performs over fifty thousand cumulative, uninterrupted hash iterations using pseudo-random memory traversal to attest a device. Attacks are detected by calculating time deviations in normal and attack scenarios. However, such an approach is complex, has high latency ($\sim$ minutes), low IoT device availability, uses precise time measurements, and requires a copy of the firmware with the verifier. SCUBA \cite{seshadri2006scuba} and SAKE \cite{seshadri2008sake} share similar problems. A few lightweight software attestation approaches use a combination of the microcontrollers' flash and RAM for firmware attestation using checksum verification but do not consider roving malware \cite{jakobsson2010retroactive,li2010sbap,chen2017secure}. Another recent study uses the SRAM directly as a feature for a Multi-layer Perceptron (MLP) classifier to attest among known malware and authentic nodes at high accuracy, however, with a limited number of classes \cite{aman2022machine}. Notable studies on hardware-based RA include \cite{krauss2007detecting,agrawal2015program,tan2011tpm}, which assume the availability of a TMP or software enclaves. This assumption does not apply to medium- and low-end IoT devices. To find a middle ground between hardware assumptions and cost-effectiveness, various hybrid approaches based on hardware/software codesign have been proposed, among which some notable studies include SMART \cite{eldefrawy2012smart}, TrustLite \cite{koeberl2014trustlite}, TyTan \cite{brasser2015tytan}, ATT-Auth \cite{aman2018att}, HAtt \cite{aman2020hatt} and \cite{hristozov2018practical}. \cite{eldefrawy2012smart,koeberl2014trustlite,brasser2015tytan,aman2018att} however, do not provide security against roving malware. Although the RAM-based approach presented in \cite{hristozov2018practical} has low latency, it requires a memory protection unit for privileged/unprivileged software execution. HAtt \cite{aman2020hatt} is another lightweight mechanism that offers a high availability of IoT nodes. It uses Physically Unclonable Functions (PUF) to protect secret information on IoT nodes from physical attacks, although requiring a copy of the microcontroller's firmware.

\subsection{Swarm attestation}

However, the feasibility and efficacy of RA techniques designed for single-node scenarios do not always relate to larger networks. This is because uniform assumptions may not apply to networks of heterogeneous devices with different communication and hardware capabilities. In addition, the attestation feature used is of great importance; for instance, leaving aside their latency, software-based approaches that use program memory checksums cannot capture relationships between nodes. These issues highlight the importance of efficient swarm attestation using suitable features; several (mostly hybrid and program memory checksum) methods have been proposed.

Notable cryptographic methods include \cite{khodari2019decentralized,asokan2015seda,carpent2017lightweight,visintin2019safe,kuang2019esdra,dushku2020sara,ibrahim2019healed}. The authors of \cite{khodari2019decentralized} propose a distributed, Merkle Hash Tree (MHT) attestation process for in-vehicle controller area networks in which. However, the nodes have high complexity and memory requirements. Furthermore, points of failure can lead to various unattested nodes. SEDA \cite{asokan2015seda} is a hybrid swarm RA technique built upon extends SMART \cite{eldefrawy2012smart} and TrustLite \cite{koeberl2014trustlite}. The verifier selects an arbitrary initiator node, and a spanning tree is created afterward. Each node calculates a hash over its memory, and the initiator collects the accumulated attestation report and shares it with the verifier for cross-checking. Despite its viability, a few limitations include its architectural impact, attestation timeout selection, and initiator node selection as highlighted by \cite{carpent2017lightweight}. The latter study proposes asynchronous and synchronous protocols called LISA$\alpha$ and LISA$s$, respectively, with minimal changes over SMART+ \cite{carpent2017lightweight}. SAFE$^d$ \cite{visintin2019safe} eliminates the need for a verifier by spreading security proofs among devices in a swarm, which they can use to validate each other. Another study, HEALED \cite{ibrahim2019healed}, can detect malicious firmware using MHT on the nodes' flash memory and disinfect compromised nodes. The five discussed studies, however, do not account for roving malware. 

Some Machine Learning (ML)-based approaches have also been proposed. WISE \cite{ammar2018wise} is an intelligent attestation method that addresses the heterogeneity of IoT devices and roving malware using a multi-clustering technique and variable attestation windows. It achieves an average detection rate of 62\%, with an evaluation latency of about 3.5 seconds. FeSA \cite{kuang2022fesa} is a distributed swarm attestation technique that uses Federated Learning (FL) and dynamic attestation periods based on the properties (device state, energy, and traffic) and security requirements for different devices, achieving an average of 87\% accuracy across various scenarios and low inference latency ($\sim$ one second). However, federated learning is typically associated with lower accuracies and the proposed method, A GNN-based distributed anomaly detection approach, is presented in \cite{protogerou2021graph}. The proposed graph MLP achieves an average 97.7\% accuracy for various network attacks using five or more seconds of network flow data. However, the model has limited applicability to malicious firmware, such as those with normal network behavior (which is a limitation of all network flow-based methods). The authors of a recent study \cite{chilese2024one} propose RAGE, a hybrid, Variational Graph Autoencoder (VGAE)-based control-flow attestation method that attests swarms using a single execution trace (or a set of executed instructions) collected from each device using a TEE. They achieve an average performance of 91\% and 98\% for return-oriented programming and data-oriented programming attacks with high availability. However, the TEE assumption does not apply to low-end devices, and the authors do not address the single point of failure in their GNN due to dropped attestation responses. The TEE assumption is also made in \cite{kibret2023property}, which uses machine learning for property-based attestation.

\subsection{Swarm-Net}

Based on the above discussion, most studies have limitations in one or more properties shown in Table \ref{tab:related}, aside from some latency and performance differences, which we highlight in Section \ref{sec:results}. There are also no available studies on SRAM-based attestation of swarms. To fill these gaps, we propose a novel GNN- and software-based attestation approach for reliable swarm attestation with high availability. \emph{Swarm-Net} makes minimal assumptions on IoT devices and verifies their state using a single SRAM trace during the attestation phase. We justify the use of SRAM over flash memory and network flow features for the following reasons: 

\begin{enumerate}
    \item SRAM is a readily available component on all medium- and low-end IoT microcontrollers.
    \item It is significantly smaller than flash memory by several orders of magnitude and is, therefore, faster to traverse.
    \item It captures runtime information, and its usage does not depend on the location of malware in the memory. It can indicate any firmware-related anomaly on the microcontroller, including roving malware.
    \item Since received messages are typically stored in the SRAM for usage, propagated anomalies are detectable.
    \item Lastly, there is no IP violation since a copy of the firmware is not required.
\end{enumerate}

Since the proposed method does not require the IoT device to compute uninterrupted checksums over the memory, it offers high availability on the IoT device. Furthermore, the interconnected nature of IoT networks is exploited using GNNs, enabling \emph{Swarm-Net} to (both) detect node-level and propagated anomalies.

\begin{table}[t]
\centering
\caption{List of relevant notations.}
\label{tab:notations}
\resizebox{0.7\columnwidth}{!}{
\begin{tabular}{|l|l|}
\hline
\textbf{Notation} & \textbf{Description}         \\ \hline
$MSE$             & Mean Square Error            \\ \hline
$CS$              & Cosine Similarity            \\ \hline
$GCN$             & Graph Convolution Network    \\ \hline
$GAT$             & Graph Attention Network      \\ \hline
$GT$              & Graph Transformer            \\ \hline
$ID_V$            & Verifier ID                  \\ \hline
$ID_G$            & Gateway ID                   \\ \hline
$ID_S$            & Swarm ID                     \\ \hline
$\text{N}_j$      & $j^{th}$ node in a swarm     \\ \hline
$T_j$             & Complete SRAM trace          \\ \hline
$T_j'$            & SRAM data section trace      \\ \hline
$b_i$             & Hex value at $i^{th}$ byte   \\ \hline
$d$               & Data section length          \\ \hline
$G_\theta$        & GNN with parameters $\theta$ \\ \hline
$\mathcal{N}_v$   & Number of one-hop neighbors  \\ \hline
$e_{ij}$          & Attention coefficients       \\ \hline
$\alpha_{ij}$     & Importance values            \\ \hline
$W$               & Weight matrix                \\ \hline
$SR$              & Swarm Response               \\ \hline
$X$               & Trainset                     \\ \hline
$x$               & Graph sample                 \\ \hline
$\hat{x}$         & Reconstructed graph          \\ \hline
$\tilde{x}$       & Graph sample with uniform noise\\ \hline
$\epsilon$        & Uniform noise                \\ \hline
$k$               & Noise factor                 \\ \hline
$DT$              & Decision Threshold           \\ \hline
$sf$              & Scaling factor               \\ \hline
$T_{def}$         & Default trace set            \\ \hline
$L$               & Maximum padding length       \\ \hline
$f$               & Predicted decision flags     \\ \hline
$C$               & Nonce                        \\ \hline
$\text{AN}_j$     & Anomaly at $\text{N}_j$      \\ \hline
$\text{S}_i$      & Simulated scenario $i$       \\ \hline
AR                & Attestation Rate             \\ \hline
DR                & Detection Rate               \\ \hline
$T_o$             & Overhead/latency             \\ \hline
\end{tabular}}
\end{table}

\begin{figure}[t]
    \centering
    \includegraphics[width=0.8\columnwidth]{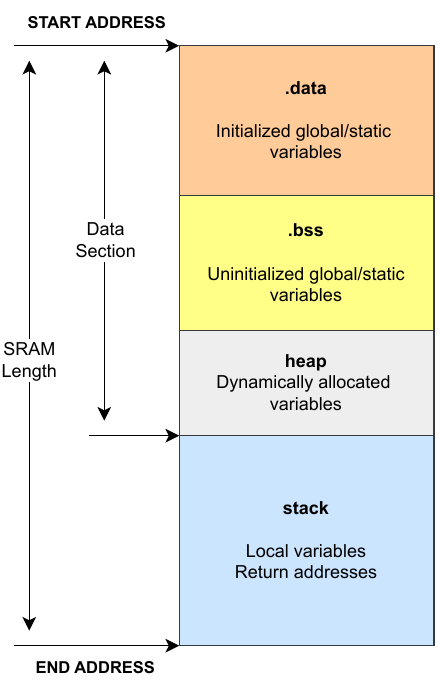}
    \caption{General organization of a microcontroller's SRAM.}
    \label{fig:sram}
\end{figure}

\section{Background}
\label{sec:background}
This section provides a basic background on SRAM and GNN; Table \ref{tab:notations} compiles a list of relevant notations for the readers' reference.

\subsection{Static Random Access Memory}
SRAM is a type of volatile memory that stores an embedded device's runtime data and is refreshed on every powerup. It is several orders of magnitude smaller than flash (thus, more practical), and its contents are refreshed on every power-up. The SRAM stores useful runtime information about the firmware installed on a microcontroller and is divided into four segments \cite{koshy2005remote} as shown in Figure \ref{fig:sram}. The .data and .bss sections store initialized and uninitialized global or static variables, respectively. The heap stores dynamically allocated variables, and the stack stores local variables and return addresses. The first three sections together may be called the data section, which has a varying size depending on the memory usage of the program loaded on the microcontroller. As we show in this paper, the data section is a useful feature for verifying firmware integrity. The stack, on the other hand, has been used for device authentication by a recent study \cite{kohli2024intelligent}. 

The data section has the same behavior for the same firmware on different devices (consistency) \cite{kohli2024intelligent} and different behavior for different firmware versions on the same device (distinguishability) \cite{aman2022machine}. Making changes to the firmware leads to changes in the SRAM during runtime, and thus, malicious code is likely to have abnormal SRAM dumps. Furthermore, in a swarm setting, data generated at one node and sent to another creates a relationship between the otherwise independent SRAM contents of the two nodes as variables with the same values are created on both devices. This can help detect the downstream effects of firmware and network anomalies. We leverage these properties to perform software-based swarm attestation of microcontroller firmware.

\textbf{SRAM notations:}
A complete SRAM trace ($t_j$) generated by an IoT device $\text{N}_j$ is represented as a sequence of $l$ bytes: 

\begin{equation}
    T_j = [b_0,b_1,...,b_{l-1}]_j\;,
\end{equation}

where $b_i$ is the hexadecimal value stored at the $i^{th}$ byte of the SRAM, and $l$ is the length of the node's SRAM. However, each node in the swarm must send only the SRAM data section: 

\begin{equation}
    T_j' = [b_0,b_1,...,b_{d-1}]_j\;,
\end{equation}

where $d$ is the size of the data section corresponding to the firmware loaded on $\text{N}_j$. Note that d << l, typically.

\subsection{Graph Neural Networks}\label{sec:GNN}
\textbf{Graph notations:}
We denote a graph $\mathcal{G}=(\mathcal{V}, \mathcal{E})$ where $\mathcal{V}$ is a set consisting of nodes $\lbrace v_1, v_2,...,v_n \rbrace$ and $\mathcal{E}\subseteq \mathcal{V}\times\mathcal{V} $ represents the set of edges in $\mathcal{G}$. For a given node $v\in\mathcal{V}$, $\mathcal{N}_v$ denotes the set of one-hop neighbors of $v$, i.e., all nodes $w\in\mathcal{N}_v$ are connected to $v$ via one edge. We also denote $\mathcal{A}\in\lbrace 0,1 \rbrace^{n\times n}$ as the adjacency matrix and $\mathcal{X}\in\mathbb{R}^{n\times d}$ as the node features of the graph. In this study, we assume a given graph is simple and unweighted.

\textbf{Network architectures:} GNNs typically follow the message-passing propagation framework, which first aggregates the features of a node $v$ with neighboring nodes $u\in\mathcal{N}_v$ using an aggregation function $\phi$, then updates them using a learnable function $f_\theta$ (usually an MLP) \cite{wu2020comprehensive}. This process is repeated several times to obtain node representations that can be used for downstream tasks. 

Edges between IoT devices can have complex relations that need to be modeled using learnable aggregation methods for a more robust understanding of how the network shares and uses information. In this paper, we consider three GNN architectures from Pytorch Geometric (PyG) \footnote{\url{https://pytorch-geometric.readthedocs.io/en/stable/}} for attestation: Graph Convolution Networks (\emph{GCN}) \cite{kipf2016semi}, Graph Attention Networks (\emph{GAT}) \cite{velivckovic2017graph} and Graph Transformers (\emph{GT}) \cite{shi2020masked}.

\emph{\underline{Graph Convolution Network}:}
Given node features $H^{(l)}$ at layer $l$ (sufficiently padded to account for bias), feature propagation in the case of the PyG implementation \emph{GCNConv} is computed as follows: 
\begin{equation}
H^{(l+1)}=\sigma(\tilde{\mathcal{D}}^{-1/2}\tilde{\mathcal{A}}\tilde{\mathcal{D}}^{-1/2}H^{(l)}W^{(l)})\;,
\end{equation}
where $\tilde{\mathcal{A}}=\mathcal{A}+\mathcal{I}$ (self-loops added), $\tilde{\mathcal{D}_{ii}}=\sum\tilde{\mathcal{A}_i}$, $W^{(l)}$ is the weight matrix for layer $l$, and $\sigma(\cdot)$ is an activation function.

Note that, GCN assigns equal importance to all neighbors \cite{velivckovic2017graph}, which may not be the best choice for the heterogeneous inter-node relationships present in IoT swarms.

\emph{\underline{Graph Attention Network}:}
\emph{GAT} computes attention coefficients $e_{ij}$ as per the following:
\begin{equation}
    e_{ij} = a(W h_i, W h_j)\;.
\end{equation}

The attention mechanism $a(\cdot)$ uses a weight matrix $W$, and the current node features $h_i$ and $h_j$ to compute the \emph{importance} of node $j$ to node $i$. These important values $\alpha_{ij}$ are normalized across the neighborhood of node $i$. The PyG implementation \emph{GATConv} computes $\alpha_{ij}$ using the following:

\begin{equation}
    \alpha_{ij} = \text{softmax}_i(e_{ij}) = \frac{\text{exp}(e_{ij})}{\sum_{k\in \mathcal{N}_i}\text{exp}(e_{ik})} \;.
\end{equation}

\emph{\underline{Graph Transformer}:}
We use the Unified Message Passing model (UniMP), which extends vanilla multi-headed attention \cite{vaswani2017attention} to graphs. The PyG implementation \emph{TransformerConv} uses $4$ weight matrices $W_1^{(l)},W_2^{(l)},W_3^{(l)},W_4^{(l)}$ per attention head at layer $l$ to compute a weighted aggregation of features of neighboring nodes. Additionally, a fifth and a sixth weight matrix ($W_5^{(l)}$ and $W_6^{(l)}$) can be introduced if gated residual connections between layers or edge features need to be considered. Attention coefficient coefficients $\alpha_{ij}$'s are computed (softmax over all nodes $j\in\mathcal{N}_i$) using

\begin{equation}
    \alpha_{ij} = \text{softmax}_i\left(\frac{q^T k_j}{\sqrt{C}}\right)\;,
\end{equation}
where $q = W_1^{(l)} h^{(l)}_i,\;k_j = W_2^{(l)} h^{(l)}_j,v_j = W_3^{(l)} h^{(l)}_j$, and feature propagation follows
\begin{equation}
    h^{(l+1)}_{i} = W_4^{(l)} h^{(l)}_i + m_j \;,
\end{equation}
where $m_j = \sum\limits_{j\in\mathcal{N}_i} \alpha_{ij} v_j$. Here $C$ is the number of output channels for a given \emph{TransformerConv} layer. For more than one head, unique $q$ and $k_j$ values are computed for each head, and the final $m_j$ is obtained by either concatenating or averaging all $m_j^{(c)}$'s.

In contrast to GAT, which uses a single shared linear transformation on every node parameterized by a single weight matrix $W$, feature propagation in GT enables learning more complex inter-node relationships due to a higher degree of parameterization using $W_{1-4}^{(l)}$ \cite{vaswani2017attention}.

\begin{figure}[t]
    \centering
    \includegraphics[width=0.9\columnwidth]{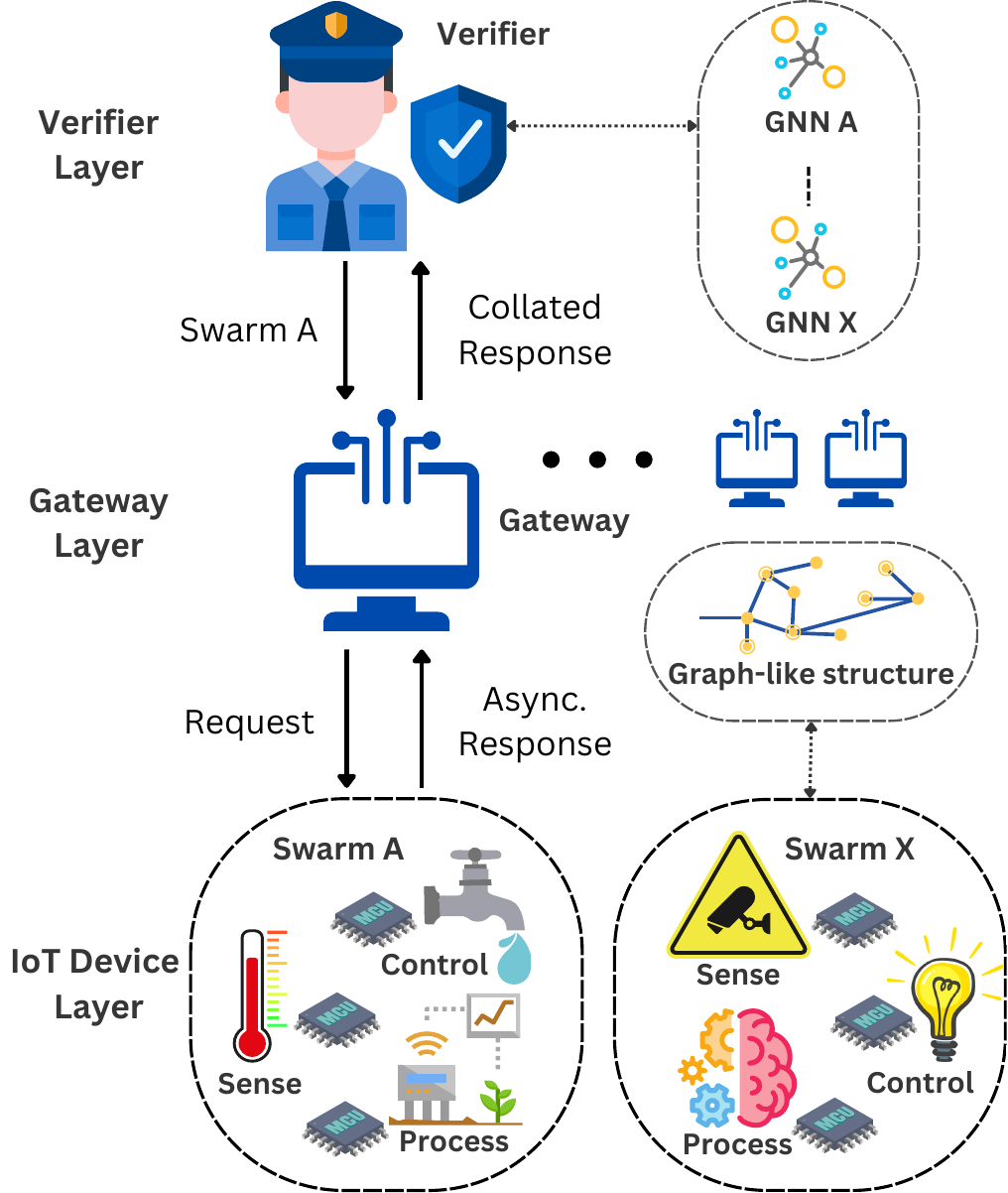}
    \caption{Layered network model.}
    \label{fig:network}
\end{figure}

\section{Network and Threat Model}
This section presents the network model and assumptions about an adversary.
\label{sec:network}
\subsection{Network Model}
Figure \ref{fig:network} shows the network model considered in this study. The network consists of three overall layers: 
\begin{enumerate}
    \item [1.] \textbf{Verifier Layer:} This layer comprises the verifier, a trusted remote device with computational power to run the proposed attestation method. The verifier sends attestation requests for a specific swarm to the gateway devices under its jurisdiction. It analyzes the collated responses from the gateways using the corresponding GNN. We assume that the communication between the verifier and the gateways is secure and that the verifier has prior knowledge of the functionality of the swarm and its constituent devices \cite{kuang2022fesa}.
    \item [2.] \textbf{Gateway Layer:} This layer consists of one or more trusted gateway devices that receive attestation requests from the verifier for the swarms under their jurisdiction, send an attestation request to the swarm specified by the verifier, collect the asynchronous responses received from a swarm's devices, and send the collated response set to the verifier.
    \item [3.] \textbf{IoT Device Layer:} This layer comprises swarms of vulnerable IoT devices, each comprising a microcontroller as the hardware platform. The tasks of the IoT devices may be divided into three overall categories (or their combinations): sense, process, and control \cite{carpent2017lightweight}. In the context of attestation, IoT devices are the provers that respond to attestation requests from the gateway with their SRAM dumps. Swarms have an inherent, graph-like structure due to the necessary communication between sense-process-control tasks. Furthermore, we assume that the direction and type of information flow in the swarm are known to the verifier, given its knowledge of the nodes' normal functionality at the time of deployment. 
\end{enumerate}

\subsection{Threat Model}
The following assumptions are made about the adversary.
\begin{enumerate}
    \item [1.] The adversary may send malicious firmware updates to one or more remote devices. In doing so, the adversary may partially or completely change an IoT device's functionality.
    \item [2.] The adversary can eavesdrop on the communication within the swarm and launch man-in-the-middle attacks. 
    \item [3.] The adversary can drop messages shared between (a) two IoT devices to create out-of-sync network states and (b) an IoT device and the gateway to cause desynchronization of the attestation round.
    \item [4.] The adversary can launch a denial of service on the swarm by impersonating the gateway and frequent attestation requests. It may achieve this by replaying previously authorized attestation requests sent to the swarm.
    \item [5.] The adversary can replay old attestation responses to impersonate a prover.
    \item [6.] The adversary cannot access the secret key or the attestation round nonce shared between the IoT microcontrollers and the trusted gateway \cite{samiullah2023group}.
\end{enumerate}

\emph{Swarm-Net} attempts to solve these threats and is discussed in Section \ref{sec:proposed}.

\begin{algorithm}[t]
\caption{Training algorithm run by $ID_V$.}
\label{algo:train}

$t \gets time()$

$request(ID_G,ID_S,samples=m)$

\While{$time()-t \geq timeout$}{
    \eIf{$receive(sender=ID_G)$}{
        $SR \gets receive(sender=ID_G)$
        
        $break()$
    }{
        $return(-1)$ \tcp{No reply from $ID_G$}
    }
}

$G_\theta \gets initializeGNN()$

$X \gets pad(int(SR)/255,L)$

$\tilde{x} = x + k\cdot\epsilon, \forall x \in X$

$e = 0$

\While{$e \leq epochs$}{
    $G_\theta \gets backprop(\theta,x,\tilde{x}), \forall x \in X$
    
    $e = e+1$
}

$DT_j \gets  sf \cdot \min_{i\in[0...m-1]} CS(X_j,\hat{X}_j), \forall N_j \in ID_S$

$T_{def,j} = \frac{\sum X_j}{m}, \forall N_j \in ID_S$

$DT \gets [DT_0,...,DT_n-1]$

$T_{def} \gets [T_{def,0},...,T_{def,n-1}]$

$saveParams(\{G_\theta,DT,T_{def},L\}_S)$

\end{algorithm}

\begin{algorithm}[t]
\caption{Attestation algorithm run by $ID_V$.}
\label{algo:attest}

$t \gets time()$

$request(ID_G,ID_S,samples=1)$

\While{$time()-t \geq timeout$}{
    \eIf{$receive(sender=ID_G)$}{
        $SR \gets receive(sender=ID_G)$
        
        $break()$
    }{
        $return(-1)$ \tcp{No reply from $ID_G$}
    }
}

$\{G_\theta,DT,T_{def},L\}_S \gets loadParams(ID_S)$

$x \gets SR == 0 \;?\; T_{def}:SR$

$x \gets padding(int(x)/255,L)$

$\hat{x} \gets G_\theta(x)$

$f \gets CS(x,\hat{x}) > DT \;?\; 0:1$

$return(f)$
\end{algorithm}

\begin{figure*}[t]
    \centering
    \includegraphics[width=0.75\textwidth]{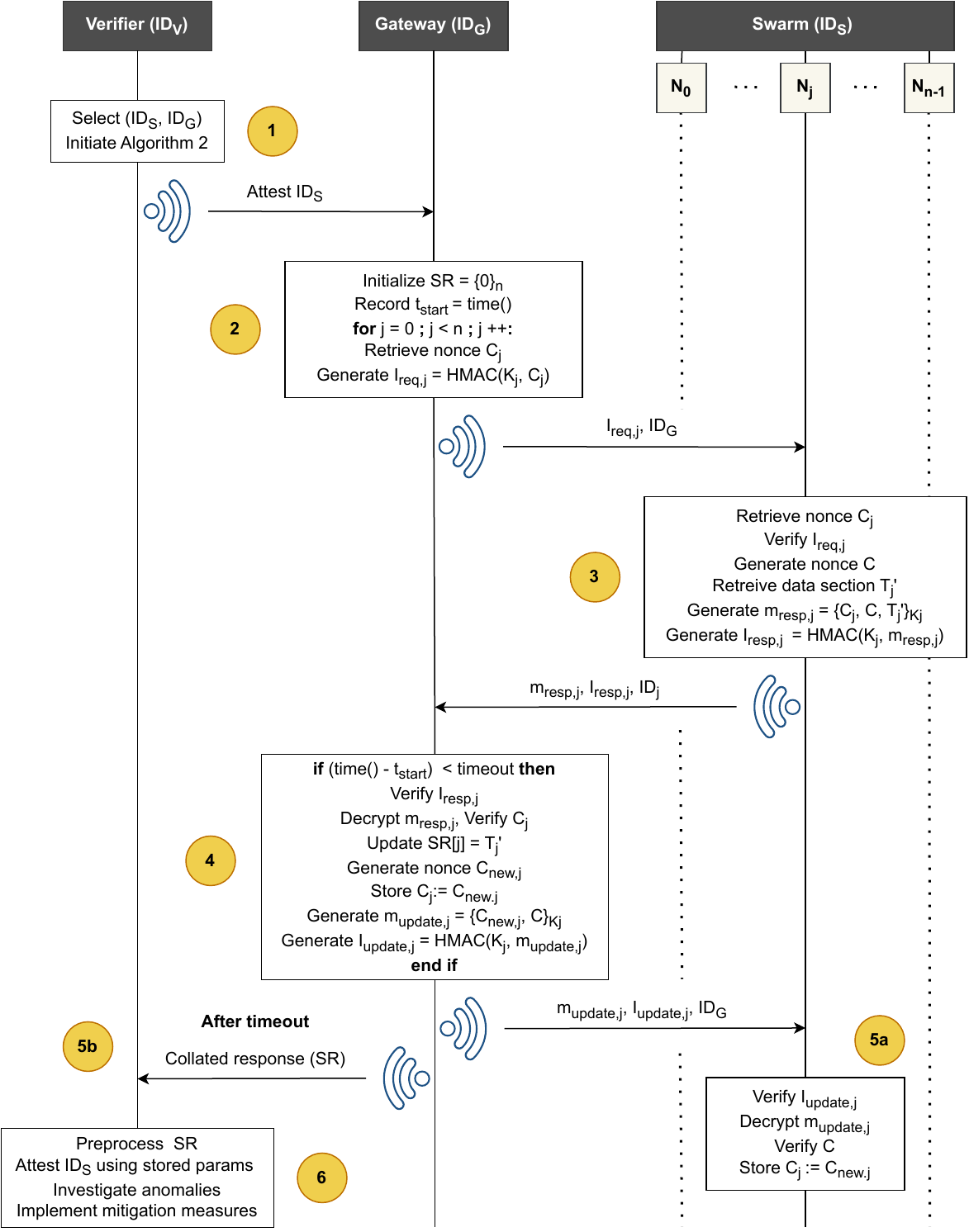}
    \caption{The \emph{Swarm-Net} attestation protocol.}
    \label{fig:protocol}
\end{figure*}

\section{Proposed Technique: Swarm-Net}
\label{sec:proposed}
This section discusses the training phase, the attestation algorithm, and the overall attestation protocol.

\subsection{Training Phase}\label{sec:training}

After deploying a swarm $ID_S$, the verifier $ID_V$ initiates the training phase depicted in Algorithm \ref{algo:train}. It prompts the gateway $ID_G$ to collect a sufficient number ($m \gtrsim 500$ based on experiments) of collated samples from each IoT device $\text{N}_j$ under normal operation. $ID_G$ collects the necessary $m$-sample swarm response $SR$ and sends it to $ID_V$. We assume this collection happens under monitored conditions with no malicious activity to ensure a clean training dataset. On the verifier's end, each $T'$ in $SR$ is converted to its integer equivalent, scaled by a factor of 255 (to bring it in a [0,1] range), and padded to an appropriately selected maximum length $L$ for uniformity during evaluation. A training set of graphs $X\in\mathbb{R}^{m \times n \times L} \subset \mathbb{X}$ is thus created, where $\mathbb{X}$ is the complete distribution of graphs associated with normal firmware in the swarm. The learnable parameters $\theta$ are then optimized using back-propagation to approximately reconstruct the training distribution $\mathbb{X}$. To do so all $x \in X$ are perturbed with scaled Uniform noise $\epsilon\sim\mathcal{U}(0,1)$ such that $\tilde{x} = x + k\cdot\epsilon$ where $\epsilon,x\in\mathbb{R}^{n\times L}$ and the GNN is trained to denoise $\tilde{x}$ such that $MSE(x,\hat{x})$ is minimized. Here, $k$ is the noise factor, $\hat{x} = G_\theta(\tilde{x})$ and $MSE(x,\hat{x})$ is computed as follows:

\begin{equation}
    MSE(x,\hat{x}) = \frac{||x-\hat{x}||^2_2}{nL} \;,
\end{equation}
where $||\cdot||_2$ is the $L_2$ norm. Adding $\epsilon$ is necessary to force $G_\theta$ to learn the output distribution $\mathbb{X}$ regardless of the input, thereby improving the model's performance on anomalous samples. Upon completion of model training, $ID_V$ computes the detection threshold $DT_j$ of each device $\text{N}_j$. While most anomaly detection studies use MSE-based thresholds, we observe that Cosine Similarity (CS) enables a more intuitive and generalized threshold selection approach. Each $DT_j$ is computed as follows:  

\begin{equation}
    DT_j = sf \cdot \min_{i\in[0...m-1]} CS(X_j,\hat{X_j}) \;,
\end{equation}
where $X_j$ is the preprocessed data of $\text{N}_j$, and $sf$ is the scaling factor of the thresholds, which decides how close the decision boundary should be to the normal CS scores. We use $sf = 0.999$ in our experiments. Further, $\hat{x} = G_\theta(x)$ and CS between two vectors $\Bar{u}$ and $\Bar{w}$ is evaluated as:

\begin{equation}
    CS(\Bar{u},\Bar{w}) = \frac{\Bar{u}\cdot\Bar{w}}{|\Bar{u}||\Bar{w}|} \;.
\end{equation}

$ID_V$ compiles each $DT_j$ into a single set $DT \in \mathbb{R}^{n\times 1}$. Further, for the event of dropped attestation responses from nodes in $ID_S$, $ID_V$ computes a set $T_{def} = [T_{def,0}...T_{def,n-1}] \in \mathbb{R}^{n\times L}$ of default traces where $T_{def,j} \in \mathbb{R}^{1 \times L}$ is computed as:

\begin{equation}
    T_{def,j} = \frac{\sum X_j}{m} \;.
\end{equation}

Here, $X_j$ is the training data of $\text{N}_j$. These average traces are in the event of dropped responses during attestation. This way, $ID_V$ can verify the firmware of other respondents while avoiding a single point of failure. Finally, $ID_V$ stores the parameter set $param = \{G_\theta, DT, T_{def}, L\}_S$ in its memory. In addition, $ID_G$ creates a set of resynchronization nonces $L_R$, which it distributes to the swarm. $L_R$ is used by $ID_G$ to resync the attestation process in case of desynchronization attacks by an adversary.

\subsection{Attestation Phase}

The attestation phase includes two sub-parts: an algorithm and the overall protocol.

\subsubsection{\textbf{Attestation algorithm}}

Algorithm \ref{algo:attest} presents the attestation procedure run by $ID_V$ to attest $ID_S$. The verifier requests $ID_G$ for one graph SRAM sample from $ID_S$ over a secure channel and waits for $ID_G$'s collated swarm response $SR$. Upon receiving $SR$, it retrieves $param = \{G_\theta, DT, T_{def}, L\}_S$ and preprocesses $SR$. It replaces the missing responses (if any) using the respective default values of the nodes in the set $T_{def}$. It then scales the traces by a factor of 255 and pads them with zeros to a common length $L$. The model $G_\theta$ then generates a reconstructed graph $\hat{x}=G_\theta(x)$. $ID_G$ evaluates the similarity scores of each node as $CS(x,\hat{x})$ and compares them with the pre-determined $DT$. All $\text{N}_j$ with $CS_j$ above their respective $DT_j$ are flagged as authorized (0) and, otherwise, anomalous (1). The algorithm returns the decision array $f$, which the verifier uses to investigate anomalies (if any) and initiate necessary threat mitigation measures. Algorithm \ref{algo:attest} is encapsulated within the attestation protocol, which is explained next.

\begin{figure*}[t]
    \centering
    \begin{subfigure}{\columnwidth}
        \centering
        \includegraphics[width=0.7\columnwidth]{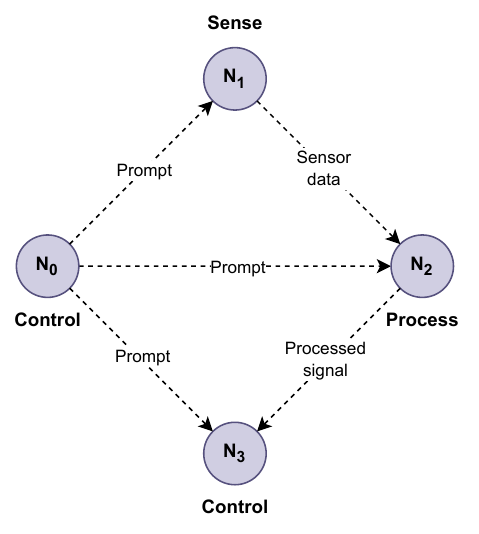}
        \caption{Four-node network configuration of Swarm-1.}
        \label{fig:4node}
    \end{subfigure}%
    \begin{subfigure}{\columnwidth}
        \centering
        \includegraphics[width=0.7\columnwidth]{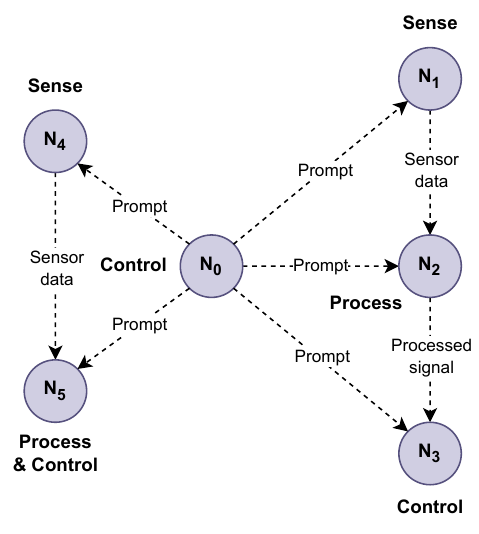}
        \caption{Six-node network configuration of Swarm-2.}
        \label{fig:6node}
    \end{subfigure}
    \caption{Swarm configurations used to collect the two datasets.}
    \label{fig:swarms}
\end{figure*}

\begin{figure}[t]
    \centering
    \includegraphics[width=0.8\columnwidth]{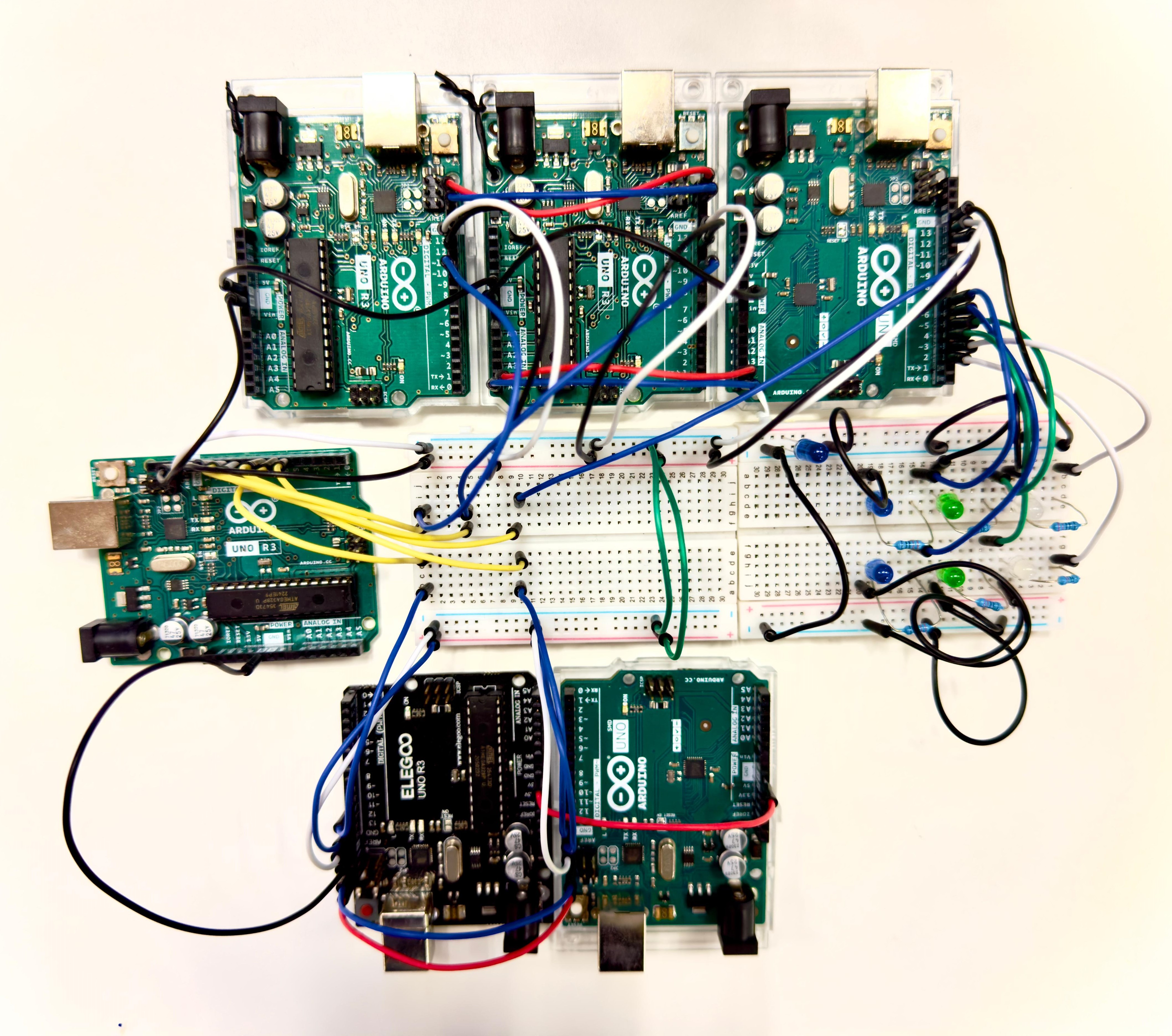}
    \caption{Physical swarm setup for dataset collection.}
    \label{fig:network}
\end{figure}

\subsubsection{\textbf{Attestation protocol}}
The \emph{Swarm-Net} attestation protocol is presented in Figure \ref{fig:protocol}. The devices use symmetric-key encryption and Hash-based Message Authentication Codes (HMAC, SHA-256 \cite{gilbert2003security}). Furthermore, we assume that the gateway securely shares the secret keys with the IoT devices during the initial setup and as needed in subsequent firmware updates \cite{samiullah2023group}. One run of the protocol consists of the six following steps:

\begin{enumerate}
    \item [1.] The verifier $ID_V$ picks a swarm $ID_S$ and its corresponding trusted gateway $ID_G$. It then initiates Algorithm \ref{algo:attest}, sends an attestation request to $ID_G$ for $ID_S$ via a secured channel and waits for a response.
    \item [2.] $ID_G$ initializes a n-size Swarm Response ($SR$) and begins the timeout timer. It retrieves a stored nonce $C_j$ and the shared secret $K_{j}$ and generates an HMAC $I_{req,j}$ for each node $\text{N}_j$ in the swarm, which it then sends to the nodes.
    \item [3.] Each $\text{N}_j$ verifies the validity of HMAC $I_{req,j}$ using the stored $C_j$ and shared secret key $K_{j}$. It generates a nonce $C$ and retrieves its SRAM data section $T_j'$. It sends an encrypted message $m_{resp,j}$ and the HMAC $I_{resp,j}$ to $ID_G$.
    \item [4.] $ID_G$ collects the responses from all nodes and accepts them if they arrive within the timeout. It verifies the accepted HMAC $I_{resp,j}$ using $K_{j}$. It decrypts $m_{resp,j}$ and verifies $C_j$. It then generates a new nonce $C_{new,j}$ for each valid respondent $\text{N}_j$, stores it into $C_j$, and creates an encrypted message $m_{update,j}$ and HMAC $I_{update,j}$ which it sends to $\text{N}_j$.
    \item [5a.] Each respondent $\text{N}_j$ verifies the received $I_{update,j}$ using $K_{j}$. It decrypts $m_{update,j}$, verifies $C$, and stores $C_{new,j}$ into $C_j$ for the next attestation round.
    \item [5b.] Upon completion of the timeout, $ID_G$ sends the collated $SR$ to $ID_V$ over a secure channel.
    \item [6.] $ID_V$ continues Algorithm \ref{algo:attest}. It preprocesses $SR$ to $x$ and evaluates the reconstructed traces $\hat{x}$ using the stored $G_\theta$. It then evaluates a trust decision $f$ for the swarm using $DT$. In case of anomalies, the administrator is prompted to investigate and take mitigation measures.
\end{enumerate}

\begin{table*}[t]
\centering
\caption{Normal and anomalous firmware in Swarm-1.}
\label{tab:swarm1}
\resizebox{\textwidth}{!}{
\begin{tabular}{|c|l|l|c|l|c|}
\hline
\textbf{Node}          & \textbf{Type}            & \multicolumn{1}{c|}{\textbf{Normal firmware functions}} & \textbf{$d$}          & \multicolumn{1}{c|}{\textbf{Anomalous firmware functions}} & \textbf{$d$}          \\ \hline
\multirow{2}{*}{$\text{N}_0$} & \multirow{2}{*}{Control} & \multirow{2}{*}{Broadcasts one byte to all nodes}       & \multirow{2}{*}{191} & Generate three random integers                             & \multirow{2}{*}{195} \\ 
                       &                          &                                                         &                      & Broadcasts one byte to all nodes                           &                      \\ \hline
\multirow{2}{*}{$\text{N}_1$} & \multirow{2}{*}{Sense}   & Generate six floating point numbers in unique ranges    & \multirow{2}{*}{450} & Generate data in an extended range                         & \multirow{2}{*}{438} \\
                       &                          & Send data to $\text{N}_2$                                      &                      & Send data to $\text{N}_2$                                         &                      \\ \hline
\multirow{2}{*}{$\text{N}_2$} & \multirow{2}{*}{Process} & Process received data into a six-byte signal            & \multirow{2}{*}{516} & Generate a random six-byte signal                          & \multirow{2}{*}{414} \\
                       &                          & Send processed signal to $\text{N}_3$                          &                      & Send control signal to $\text{N}_3$                               &                      \\ \hline
$\text{N}_3$                  & Control                  & Control six LEDs using the processed signal             & 406                  & Control six LEDs at random                                 & 386                  \\ \hline
\end{tabular}}
\end{table*}

\begin{table}[!t]
\centering
\caption{Swarm-1 dataset scenarios.}
\label{tab:dat1}
\resizebox{0.95\columnwidth}{!}{
\begin{tabular}{|c|c|c|c|c|}
\hline
\textbf{S.No.} & \textbf{Scenario} & \textbf{\begin{tabular}[c]{@{}c@{}}Primary\\ Anomaly\end{tabular}} & \textbf{\begin{tabular}[c]{@{}c@{}}Secondary\\ Anomaly\end{tabular}} & \textbf{\begin{tabular}[c]{@{}c@{}}Swarm Label\\ $N_0-...-N_3$\end{tabular}} \\ \hline
1-2 & $\text{D}_{1-2}$         & -                                                                & -                                                                      & 0-0-0-0 \\ \hline
3-4 & $\text{P}_{1-2}$         & -                                                                & -                                                                      & 0-0-0-0 \\ \hline
5   & $\text{AN}_0$            & $\text{N}_0$                                                            & -                                                                      & 1-0-0-0 \\ \hline
6   & $\text{AN}_1$            & $\text{N}_1$                                                            & $\text{N}_2$, $\text{N}_3$                                                           & 0-1-1-1 \\ \hline
7   & $\text{AN}_2$            & $\text{N}_2$                                                            & $\text{N}_3$                                                                  & 0-0-1-1 \\ \hline
8   & $\text{AN}_3$            & $\text{N}_3$                                                            & -                                                                      & 0-0-0-1 \\ \hline
9   & $\text{AN}_{12}$         & $\text{N}_1$, $\text{N}_2$                                                     & $\text{N}_3$                                                                  & 0-1-1-1 \\ \hline
10  & $\text{AN}_{23}$         & $\text{N}_2$, $\text{N}_3$                                                     & -                                                                      & 0-0-1-1 \\ \hline
11  & $\text{AN}_{13}$         & $\text{N}_1$, $\text{N}_3$                                                     & $\text{N}_2$                                                                  & 0-1-1-1 \\ \hline
12  & $\text{AN}_{123}$        & $N_{1-3}$                                                        & $\text{N}_2$, $\text{N}_3$                                                           & 0-1-1-1 \\ \hline
13  & $\text{AN}_{0123}$       & $N_{0-3}$                                                        & $\text{N}_2$, $\text{N}_3$                                                           & 1-1-1-1 \\ \hline
\end{tabular}}
\end{table}

\begin{table*}[t]
\centering
\caption{Normal and anomalous firmware in Swarm-2.}
\label{tab:swarm2}
\resizebox{\textwidth}{!}{
\begin{tabular}{|c|l|l|c|l|c|}
\hline
\textbf{Node}          & \textbf{Type}                                                                 & \multicolumn{1}{c|}{\textbf{Normal firmware functions}} & \textbf{$d$}          & \multicolumn{1}{c|}{\textbf{Anomalous firmware functions}} & \textbf{$d$}          \\ \hline
\multirow{2}{*}{$\text{N}_0$} & \multirow{2}{*}{Control}                                                      & \multirow{2}{*}{Broadcasts one byte to all nodes}       & \multirow{2}{*}{195} & Generate two random integers                               & \multirow{2}{*}{199} \\
                       &                                                                               &                                                         &                      & Broadcasts one byte to all nodes                           &                      \\ \hline
\multirow{2}{*}{$\text{N}_1$} & \multirow{2}{*}{Sense}                                                        & Generate four floating point numbers in unique ranges   & \multirow{2}{*}{438} & Generate data in an extended range                         & \multirow{2}{*}{430} \\
                       &                                                                               & Send data to $\text{N}_2$                                      &                      & Send data to $\text{N}_2$                                         &                      \\ \hline
\multirow{2}{*}{$\text{N}_2$} & \multirow{2}{*}{Process}                                                      & Process received data into a four-byte signal           & \multirow{2}{*}{490} & Generate a random four-byte signal                         & \multirow{2}{*}{414} \\
                       &                                                                               & Send processed signal to $\text{N}_3$                          &                      & Send control signal to $\text{N}_3$                               &                      \\ \hline
$\text{N}_3$                  & Control                                                                       & Control four LEDs using the processed signal            & 394                  & Control four LEDs at random                                & 386                  \\ \hline
\multirow{2}{*}{$\text{N}_4$} & \multirow{2}{*}{Sense}                                                        & Generate three floating point numbers in unique ranges  & \multirow{2}{*}{430} & Generate data normally                                     & \multirow{2}{*}{372} \\
                       &                                                                               & Send data to $\text{N}_5$                                      &                      & Does not send data to $\text{N}_5$                                &                      \\ \hline
\multirow{2}{*}{$\text{N}_5$} & \multirow{2}{*}{\begin{tabular}[l]{@{}l@{}}Process \&\\ Control\end{tabular}} & Process received data into a three-byte signal          & \multirow{2}{*}{446} & Process received data normally                             & \multirow{2}{*}{452} \\
                       &                                                                               & Control three LEDs using the processed signal           &                      & Control three unauthorized LEDs                      &                      \\ \hline
\end{tabular}}
\end{table*}

\begin{table}[t]
\centering
\caption{Swarm-2 dataset scenarios.}
\label{tab:dat2}
\resizebox{0.9\columnwidth}{!}{
\begin{tabular}{|c|c|c|c|c|}
\hline
\textbf{S.No.} & \textbf{Scenario} & \textbf{\begin{tabular}[c]{@{}c@{}}Primary\\ Anomaly\end{tabular}} & \textbf{\begin{tabular}[c]{@{}c@{}}Secondary\\ Anomaly\end{tabular}} & \textbf{\begin{tabular}[c]{@{}c@{}}Swarm Label\\ $N_0-...-N_5$\end{tabular}} \\ \hline
1-4 & $\text{D}_{1-4}$         & -                                                                  & -                                                                    & 0-0-0-0-0-0 \\ \hline
5   & $\text{AN}_0$            & $\text{N}_0$                                                              & -                                                                    & 1-0-0-0-0-0 \\ \hline
6   & $\text{AN}_1$            & $\text{N}_1$                                                              & $\text{N}_2$, $\text{N}_3$                                                         & 0-1-1-1-0-0 \\ \hline
7   & $\text{AN}_2$            & $\text{N}_2$                                                              & $\text{N}_3$                                                                & 0-0-1-1-0-0 \\ \hline
8   & $\text{AN}_3$            & $\text{N}_3$                                                              & -                                                                    & 0-0-0-1-0-0 \\ \hline
9   & $\text{AN}_4$            & $\text{N}_4$                                                              & $\text{N}_5$                                                                & 0-0-0-0-1-1 \\ \hline
10  & $\text{AN}_5$            & $\text{N}_5$                                                              & -                                                                    & 0-0-0-0-0-1 \\ \hline
\end{tabular}}
\end{table}

\section{Experimental Setup}
This section covers the experimental setup used in this paper, including the devices, programming platforms and libraries, sample IoT swarms, datasets, hyperparameters, and evaluation metrics.
\label{sec:setup}

\subsection{Hardware and Software}

Figure \ref{fig:swarms} and \ref{fig:network} show sample swarms and the physical setup used for dataset collection, respectively. All IoT nodes are variants of the Arduino UNO Rev3 and Elegoo UNO Rev3 devices, each with an 8-bit ATmega328P microcontroller and a 2KB SRAM. The gateway device is a Dell Latitude laptop with an Intel i7 processor and 16GB DRAM. Further, we use two verifier devices: GPU experiments were run on a compute server with an AMD EPYC 7742 processor and an Nvidia A100-SXM4 GPU, and CPU experiments were run locally on a laptop with an Intel(R) Core(TM) Ultra 7 155H processor.

The firmware was programmed in Arduino IDE 2.3.2, while the data collection and attestation codes were written in Python 3.8. The main Python libraries used are Numpy 1.26.4, Pandas 2.2.2, Pyserial 3.5, and Pytorch Geometric 2.5.3.

\subsection{IoT Swarms, Datasets, and Simulations}
We present our SRAM dataset on IEEE Dataport \cite{gmee-vj41-24}, collected from the swarms shown in Figure \ref{fig:swarms}. In addition, we simulate three network attack scenarios to test the holistic robustness of \emph{Swarm-Net}.

\subsubsection{\textbf{Swarm-1}}
Swarm-1 is the four-node network shown in Figure \ref{fig:4node}. 

\textit{Behaviors:} This swarm captures normal behavior, physical twins of development networks, malicious firmware, faulty/tampered data, and abnormal peripheral control. 

\textit{Firmware:} Details about the normal and anomalous variants of each node are provided in Table \ref{tab:swarm1}. $\text{N}_0$ sends one byte to all other nodes in the swarm. In its anomalous variant, $\text{N}_0$ generates three random integers (six bytes) in addition to its main function. The anomalous variant doesn't affect other nodes, and its purpose is to detect node-level anomalies with the same network behavior as the normal variant, as well as to check if the GNN misclassifies other nodes due to message passing. $\text{N}_1$ is a sense-type node that generates six floating point numbers (twenty-four bytes) in unique ranges and sends them to $\text{N}_2$. In its anomalous variant, $\text{N}_1$ generates these numbers in an extended range (including the original), creating partially/completely faulty data collection scenarios at downstream nodes ($\text{N}_2$ and $\text{N}_3$). $\text{N}_2$ is a process-type node that receives six floating point numbers (twenty-four bytes) and generates a control signal (six bytes) based on pre-defined logic. It then sends this processed signal to $\text{N}_3$. In its anomalous variant, $\text{N}_2$ discards the data received from $\text{N}_1$ and instead generates a random control signal for $\text{N}_3$, which affects $\text{N}_3$'s operation. Lastly, $\text{N}_3$ is a control-type node that uses the processed signals from $\text{N}_2$ to control six LEDs. In its anomalous variant, $\text{N}_3$ discards the received signal and instead controls the output LEDs at random.

\textit{Dataset:} The Swarm-1 dataset includes thirteen scenarios (four normal and nine anomalous), as shown in Table \ref{tab:dat1}. $\text{D}_1$ and $\text{D}_2$ are development sets used for training and testing and are collected from two independent initializations of the original swarm network. In contrast, $\text{P}_1$ and $\text{P}_2$ are collected from two physical twin networks made using a different set of devices. Among the anomalous scenarios, $\text{AN}_j$ refers to an anomaly introduced at $\text{N}_j$. Each scenario has 400 synchronized, sequential SRAM traces for each $\text{N}_j$, where each trace is the sequence of integer equivalents of the hex content of the SRAM. 

\subsubsection{\textbf{Swarm-2}}
Swarm-2 is the six-node network shown in Figure \ref{fig:6node} that offers a greater challenge to the verifier given its more complex design, a lower number of communicated bytes, a lower degree of change in firmware for anomalous cases, and a larger number of threat types.

\textit{Behaviors:} This swarm captures normal behavior, malicious firmware, faulty/tampered data, dropped messages, out-of-sync states, abnormal peripheral control, and additional malicious peripherals.

\textit{Firmware:} The details of the normal and anomalous firmware of each node are given in Table \ref{tab:swarm2}. $\text{N}_0$ has the same purpose as explained in Swarm-1. It sends one byte to $N_1-N_5$. In its anomalous variant, $\text{N}_0$ generates two random integers (four bytes). The first branch in this swarm consists of three nodes: $\text{N}_1$ is a process-type node that generates four floating point numbers (sixteen bytes) in unique ranges. Its anomalous variant generates random data in extended ranges, causing downstream effects at $\text{N}_2$ and $\text{N}_3$. $\text{N}_2$ is a process-type node that uses the data received from $\text{N}_1$, and generates a four-byte processed signal for $\text{N}_3$. In its anomalous variant, $\text{N}_2$ sends a random processed signal to $\text{N}_3$, affecting its operation. $\text{N}_3$ is a control-type node that uses the received processed signal to control four LEDs. In its anomalous variant, $\text{N}_3$ generates a random control signal. The second branch of Swarm-2 comprises two nodes: $\text{N}_4$ is a process-type node that generates three floating point numbers (twelve bytes) and sends them to $\text{N}_5$. In its anomalous variant, the function used to send data is commented out (tampering with control dependencies). This anomaly creates an out-of-sync state between $\text{N}_4$ and $\text{N}_5$ and should be detectable at both nodes during attestation. Lastly, $\text{N}_5$ is a process- and control-type node that processes the data received from $\text{N}_4$ and controls three LEDs. In its anomalous variant, $\text{N}_5$ uses the received data to control three LEDs connected to different microcontroller pins.

\textit{Dataset:} The Swarm-2 dataset consists of ten scenarios (four normal and six anomalous), as shown in Table \ref{tab:dat2}. $\text{D}_1-\text{D}_4$ are development sets for training and testing collected from four separate initializations of the development network. $\text{AN}_0-\text{AN}_5$ are six anomalous scenarios where $\text{AN}_j$ corresponds to an anomaly at $\text{N}_j$. Each scenario comprises 900 synchronized, sequential SRAM traces for all $\text{N}_j$.

\begin{table}[t]
\centering
\caption{Behavior types encapsulated in the swarm datasets and simulated scenarios.}
\label{tab:categories}
\resizebox{0.9\columnwidth}{!}{
\begin{tabular}{|c|l|l|}
\hline
\textbf{S.No.} & \textbf{Behavior type}      & \textbf{Scenarios}                                                                                        \\ \hline
1              & Normal behavior             & $\text{D}_{1-4}$, $\text{P}_{1-2}$                                                                                          \\ \hline
2              & Physical twins               & $\text{P}_{1-2}$                                                                                                           \\ \hline
3              & Malicious firmware          & $\forall\text{AN}_i$                                                                                                          \\ \hline
4              & Propagated anomalies        & \begin{tabular}[c]{@{}l@{}}$\text{AN}_1$, $\text{AN}_2$, $\text{AN}_{12}$, \\ $\text{AN}_{13}$, $\text{AN}_4$\end{tabular} \\ \hline
5              & Faulty/tampered data        & $\text{AN}_1$, $\text{AN}_2$                                                                                               \\ \hline
6              & Abnormal peripheral control & $\text{AN}_3$, $\text{AN}_5$                                                                                                              \\ \hline
7              & Tampered functions          & $\text{AN}_4$                                                                                                              \\ \hline
8              & Out-of-sync states          & $\text{AN}_4$, $\text{S}_3$                                                                                                              \\ \hline
9              & Added peripherals           & $\text{AN}_5$                                                                                                              \\ \hline
10             & Dropped responses           & $\text{S}_1$                                                                                                              \\ \hline
11             & SRAM perturbation           & $\text{S}_2$                                                                                                              \\ \hline
12             & Trace replay                & $\text{S}_3$                                                                                                              \\ \hline
\end{tabular}}
\end{table}

\begin{table*}[t]
\centering
\caption{Attestation results on Swarm-1 averaged over twenty training and testing phases.}
\label{tab:results1}
\resizebox{0.8\textwidth}{!}{
\begin{tabular}{|c|cccc|cccc|cccc|} \hline
                                    & \multicolumn{4}{c|}{\textbf{GCN}}                                                         & \multicolumn{4}{c|}{\textbf{GAT}}                                                        & \multicolumn{4}{c|}{\textbf{GT}}                                                                              \\ \cline{2-13}
\multirow{-2}{*}{\textbf{Scenario}} & \textbf{$\text{N}_0$} & \textbf{$\text{N}_1$} & \textbf{$\text{N}_2$}               & \textbf{$\text{N}_3$}               & \textbf{$\text{N}_0$} & \textbf{$\text{N}_1$} & \textbf{$\text{N}_2$}              & \textbf{$\text{N}_3$}               & \textbf{$\text{N}_0$}               & \textbf{$\text{N}_1$}               & \textbf{$\text{N}_2$}               & \textbf{$\text{N}_3$}               \\ \hline
$\text{D}_1$                        & 100            & 100            & 100                          & 100                          & 100            & 100            & 100                         & 100                          & {\color[HTML]{036400} 100}   & {\color[HTML]{036400} 100}   & {\color[HTML]{036400} 100}   & {\color[HTML]{036400} 100}   \\
$\text{D}_2$                        & 100            & 100            & 100                          & 100                          & 100            & 100            & 100                         & 100                          & {\color[HTML]{036400} 100}   & {\color[HTML]{036400} 100}   & {\color[HTML]{036400} 100}   & {\color[HTML]{036400} 100}   \\
$\text{P}_1$                        & 99.67          & 99.67          & 100                          & 100                          & 99.67          & 99.67          & 100                         & 100                          & {\color[HTML]{036400} 100}   & {\color[HTML]{036400} 100}   & {\color[HTML]{036400} 100}   & {\color[HTML]{036400} 100}   \\
$\text{P}_2$                        & 100            & 99.33          & 99                           & 100                          & 100            & 99.33          & 99                          & 100                          & {\color[HTML]{036400} 100}   & 98.63 & {\color[HTML]{036400} 99.33} & {\color[HTML]{036400} 100}   \\
$\text{AN}_0$                       & 100            & 100            & 97.95                        & 100                          & 100            & 100            & 100                         & 100                          & {\color[HTML]{036400} 100}   & 99.98 & {\color[HTML]{036400} 100}   & {\color[HTML]{036400} 100}   \\
$\text{AN}_1$                       & 100            & 100            & {\color[HTML]{FE0000} 13.98} & 95.38                        & 100            & 100            & {\color[HTML]{FE0000} 3.68} & {\color[HTML]{FE0000} 81.87} & {\color[HTML]{036400} 100}   & {\color[HTML]{036400} 100}   & {\color[HTML]{036400} 97.48} & {\color[HTML]{036400} 99.02} \\
$\text{AN}_2$                       & 99             & 100            & 100                          & {\color[HTML]{FE0000} 20.72} & 99.33          & 100            & 100                         & {\color[HTML]{FE0000} 0.83}  & {\color[HTML]{036400} 99.63} & {\color[HTML]{036400} 100}   & {\color[HTML]{036400} 100}   & {\color[HTML]{036400} 100}   \\
$\text{AN}_3$                       & 100            & 100            & 99.67                        & 100                          & 100            & 100            & 99.67                       & 100                          & {\color[HTML]{036400} 100}   & {\color[HTML]{036400} 100}   & {\color[HTML]{036400} 100}   & {\color[HTML]{036400} 100}   \\
$\text{AN}_{12}$                    & 100            & 100            & 100                          & {\color[HTML]{FE0000} 21.48} & 100            & 100            & 100                         & {\color[HTML]{FE0000} 0.98}  & {\color[HTML]{036400} 100}   & {\color[HTML]{036400} 100}   & {\color[HTML]{036400} 100}   & {\color[HTML]{036400} 100}   \\
$\text{AN}_{23}$                    & 99.33          & 100            & 100                          & 100                          & 99.67          & 100            & 100                         & 100                          & {\color[HTML]{036400} 100}   & 99.97 & {\color[HTML]{036400} 100}   & {\color[HTML]{036400} 100}   \\
$\text{AN}_{13}$                    & 99.78          & 100            & {\color[HTML]{FE0000} 13.45} & 100                          & 100            & 100            & {\color[HTML]{FE0000} 2.33} & 100                          & {\color[HTML]{036400} 100}   & {\color[HTML]{036400} 100}   & {\color[HTML]{036400} 96.95} & {\color[HTML]{036400} 100}   \\
$\text{AN}_{123}$                   & 100            & 100            & 100                          & 100                          & 100            & 100            & 100                         & 100                          & {\color[HTML]{036400} 100}   & {\color[HTML]{036400} 100}   & {\color[HTML]{036400} 100}   & {\color[HTML]{036400} 100}   \\
$\text{AN}_{0123}$                  & 100            & 100            & 100                          & 100                          & 100            & 100            & 100                         & 100                          & {\color[HTML]{036400} 100}   & {\color[HTML]{036400} 100}   & {\color[HTML]{036400} 100}   & {\color[HTML]{036400} 100} \\ \hline
\end{tabular}}
\end{table*}

\begin{table*}[t]
\centering
\caption{Attestation results on Swarm-2 averaged over twenty training and testing phases.}
\label{tab:results2}
\resizebox{\textwidth}{!}{
\begin{tabular}{|c|cccccc|cccccc|cccccc|} \hline
                                    & \multicolumn{6}{c|}{\textbf{GCN}}                                                                                                     & \multicolumn{6}{c|}{\textbf{GAT}}                                                                                                     & \multicolumn{6}{c|}{\textbf{GT}}                                                                                                                                          \\ \cline{2-19}
\multirow{-2}{*}{\textbf{Scenario}} & \textbf{$\text{N}_0$}               & \textbf{$\text{N}_1$} & \textbf{$\text{N}_2$}           & \textbf{$\text{N}_3$}               & \textbf{$\text{N}_4$} & \textbf{$\text{N}_5$} & \textbf{$\text{N}_0$}               & \textbf{$\text{N}_1$} & \textbf{$\text{N}_2$}           & \textbf{$\text{N}_3$}               & \textbf{$\text{N}_4$} & \textbf{$\text{N}_5$} & \textbf{$\text{N}_0$}               & \textbf{$\text{N}_1$}               & \textbf{$\text{N}_2$}               & \textbf{$\text{N}_3$}               & \textbf{$\text{N}_4$}               & \textbf{$\text{N}_5$}             \\ \hline
$\text{D}_1$             & 100                          & 100            & 100                      & 100                          & 100            & 100            & 100                          & 100            & 100                      & 100                          & 100            & 100            & {\color[HTML]{036400} 100}   & {\color[HTML]{036400} 100}   & {\color[HTML]{036400} 100}   & {\color[HTML]{036400} 100}   & {\color[HTML]{036400} 100}   & {\color[HTML]{036400} 100} \\
$\text{D}_2$               & 100                          & 100            & 100                      & 100                          & 100            & 100            & 100                          & 100            & 100                      & 100                          & 100            & 100            & {\color[HTML]{036400} 100}   & {\color[HTML]{036400} 100}   & {\color[HTML]{036400} 100}   & {\color[HTML]{036400} 100}   & {\color[HTML]{036400} 100}   & {\color[HTML]{036400} 100} \\
$\text{D}_3$                        & 100                          & 99.88          & 100                      & 100                          & 99.88          & 100            & 100                          & 99.88          & 100                      & 100                          & 99.88          & 100            & {\color[HTML]{036400} 100}   & {\color[HTML]{036400} 99.88} & {\color[HTML]{036400} 100}   & {\color[HTML]{036400} 100}   & {\color[HTML]{036400} 100}   & {\color[HTML]{036400} 100} \\
$\text{D}_4$                        & 100                          & 99.75          & 100                      & 100                          & 100            & 100            & 100                          & 99.75          & 100                      & 100                          & 100            & 100            & {\color[HTML]{036400} 100}   & {\color[HTML]{036400} 99.88} & {\color[HTML]{036400} 100}   & {\color[HTML]{036400} 100}   & {\color[HTML]{036400} 100}   & {\color[HTML]{036400} 100} \\
$\text{AN}_0$                       & {\color[HTML]{FE0000} 71.03} & 99.99          & 100                      & 100                          & 100            & 100            & {\color[HTML]{FE0000} 12.42} & 99.75          & 99.94                    & 99.88                        & 100            & 100            & {\color[HTML]{036400} 100}   & {\color[HTML]{036400} 99.99} & {\color[HTML]{036400} 100}   & {\color[HTML]{036400} 100}   & 99.88 & {\color[HTML]{036400} 100} \\
$\text{AN}_1$                       & 100                          & 100            & {\color[HTML]{FE0000} 0} & {\color[HTML]{FE0000} 1.12}  & 100            & 100            & 100                          & 100            & {\color[HTML]{FE0000} 0} & {\color[HTML]{FE0000} 2.16}  & 100            & 100            & 99.99 & {\color[HTML]{036400} 100}   & {\color[HTML]{036400} 100}   & {\color[HTML]{036400} 98.08} & {\color[HTML]{036400} 100}   & {\color[HTML]{036400} 100} \\
$\text{AN}_2$                       & 100                          & 100            & 100                      & {\color[HTML]{FE0000} 36.02} & 100            & 100            & 99.97                        & 100            & 100                      & {\color[HTML]{FE0000} 41.69} & 100            & 100            & 99.98 & 99.88 & {\color[HTML]{036400} 100}   & {\color[HTML]{036400} 100}   & {\color[HTML]{036400} 100}   & {\color[HTML]{036400} 100} \\
$\text{AN}_3$                       & 100                          & 100            & 100                      & 100                          & 100            & 100            & 100                          & 99.88          & 100                      & 100                          & 100            & 100            & {\color[HTML]{036400} 100}   & {\color[HTML]{036400} 100}   & {\color[HTML]{036400} 100}   & {\color[HTML]{036400} 100}   & {\color[HTML]{036400} 100}   & {\color[HTML]{036400} 100} \\
$\text{AN}_4$                       & 99.75                        & 100            & 100                      & 100                          & 100            & 100            & 100                          & 100            & 100                      & 100                          & 100            & 100            & {\color[HTML]{036400} 100}   & {\color[HTML]{036400} 100}   & {\color[HTML]{036400} 100}   & {\color[HTML]{036400} 100}   & {\color[HTML]{036400} 100}   & {\color[HTML]{036400} 100} \\
$\text{AN}_5$                       & 100                          & 100            & 100                      & 100                          & 100            & 100            & 100                          & 100            & 100                      & 100                          & 100            & 100            & {\color[HTML]{036400} 100}   & {\color[HTML]{036400} 100}   & 99.99 & {\color[HTML]{036400} 100}   & {\color[HTML]{036400} 100}   & {\color[HTML]{036400} 100} \\ \hline
\end{tabular}}
\end{table*}

\subsubsection{\textbf{Simulated Scenarios}}
In addition to the scenarios in the Swarm-1 and Swarm-2 datasets, we use the development sets to simulate the following scenarios:

\begin{enumerate}
    \item [1.] \underline{Dropped responses ($\text{S}_1$)}: Since adversaries may drop attestation responses from the swarm to the gateway, it is important to evaluate the performance of Algorithm \ref{algo:attest} in such a case. To simulate dropped responses, attestation responses from randomly selected $\text{N}_j$ are dropped and replaced with the corresponding default trace $t_{def,j}$ from $T_{def}$ during attestation.
    \item [2.] \underline{SRAM perturbation ($\text{S}_2$)}: An adversary may attempt to forge an SRAM trace or attack the integrity of communication by changing the messages sent from the swarm to the gateway. To simulate such a scenario, we add varying bytes of randomness to the development sets. 
    \item [3.] \underline{Trace replay/out-of-sync states ($\text{S}_3$)}: An adversary may drop messages between nodes or attempt to replay attestation responses collected from nodes. In doing so, they may create network states that are not in sync with the current state of the Sense-Process-Control tasks in the network. To create such a threat scenario, we randomly shuffle the development sets along the temporal axis.
\end{enumerate}

\subsection{Models and Hyperparameters}
We experiment with three GNN architectures - GCN, GAT, and GT using the respective \emph{GCNConv}, \emph{GATConv}, and \emph{TransformerConv} layers from PyG. All models consist of two graph layers compressing input data to a latent dimension of $32$, followed by a single Linear layer for reconstruction. In all experiments, we use the Adam \cite{kingma2014adam} optimizer with a learning rate of $0.01$ and a weight decay ($L2$ regularization on model weights) of $0.0005$. Uniform noise sampled from $\mathcal{U}(0,1)$ are added to the inputs with a coefficient of $0.4$.

\subsection{Evaluation Metrics}  
The proposed threshold-based anomaly detection approach is essentially a binary classification problem, wherein similarity scores of SRAM data sections lying above their respective thresholds are flagged as safe (0) and otherwise anomalous (1). We define True Positive ($TP$), True Negative ($TN$), False Positive ($FP$), and False Negative ($FN$) as the correct 1, correct 0, incorrect 1, and 0 invalid flags after thresholding, respectively.

Subsequently, we use three metrics: Accuracy ($A$), Detection Rate ($DR$), and Attestation Rate ($AR$) defined as follows:

\begin{equation}
\label{eq:acc}
    A = \frac{TP+TN}{TP+TN+FP+FN}
\end{equation}

\begin{equation}
    DR = \frac{TP}{TP+FN}
\end{equation}

\begin{equation}
    AR = \frac{TN}{TN+FP}
\end{equation}

Here, $DR$ and $AR$ are analogous to the true positive and negative rates, respectively.

\section{Results}
\label{sec:results}
This section presents the experimental results observed in various test scenarios (categorized in Table \ref{tab:categories}), an analysis of performance, and a quantitative comparison with related works.

\subsection{Swarm-1}
Table \ref{tab:results1} shows the detection results on the Swarm-1 dataset. In addition, Figure \ref{fig:results1} compiles these results into four categories: overall performance, authentic firmware, anomalous firmware, and propagated anomalies, and compares the three proposed model architectures.

As the table and figure show, GT has an overall accuracy of 99.83\% with a 99.92\% AR on authentic firmware, 100\% DR on anomalous firmware, and 98.7\% DR on propagated anomalies. While GAT and GCN have comparable performance to GT in normal and anomalous firmware samples, their performance is significantly lower on propagated anomalies observed in $\text{AN}_1$, $\text{AN}_2$, $\text{AN}_{12}$ and $\text{AN}_{13}$. In such cases, GAT and GCN attain an average DR of 18\% and 33\%, lowering their overall accuracy to 92.04\% and 93.43\%, respectively.

\textit{\textbf{Intuition:}} Based on our discussion on the capabilities of different aggregation methods in Section \ref{sec:GNN}, GT has an advantage over GAT and GCN regarding its ability to learn more complex network edge behaviors. This is supported by the results observed for Swarm-1 and Swarm-2 (in the next subsection). This reasoning can also extend to other GNN and ML architectures.

\begin{figure*}[t]
    \centering
    \begin{subfigure}{\columnwidth}
        \centering
        \includegraphics[trim={0 80 0 80},width=\columnwidth]{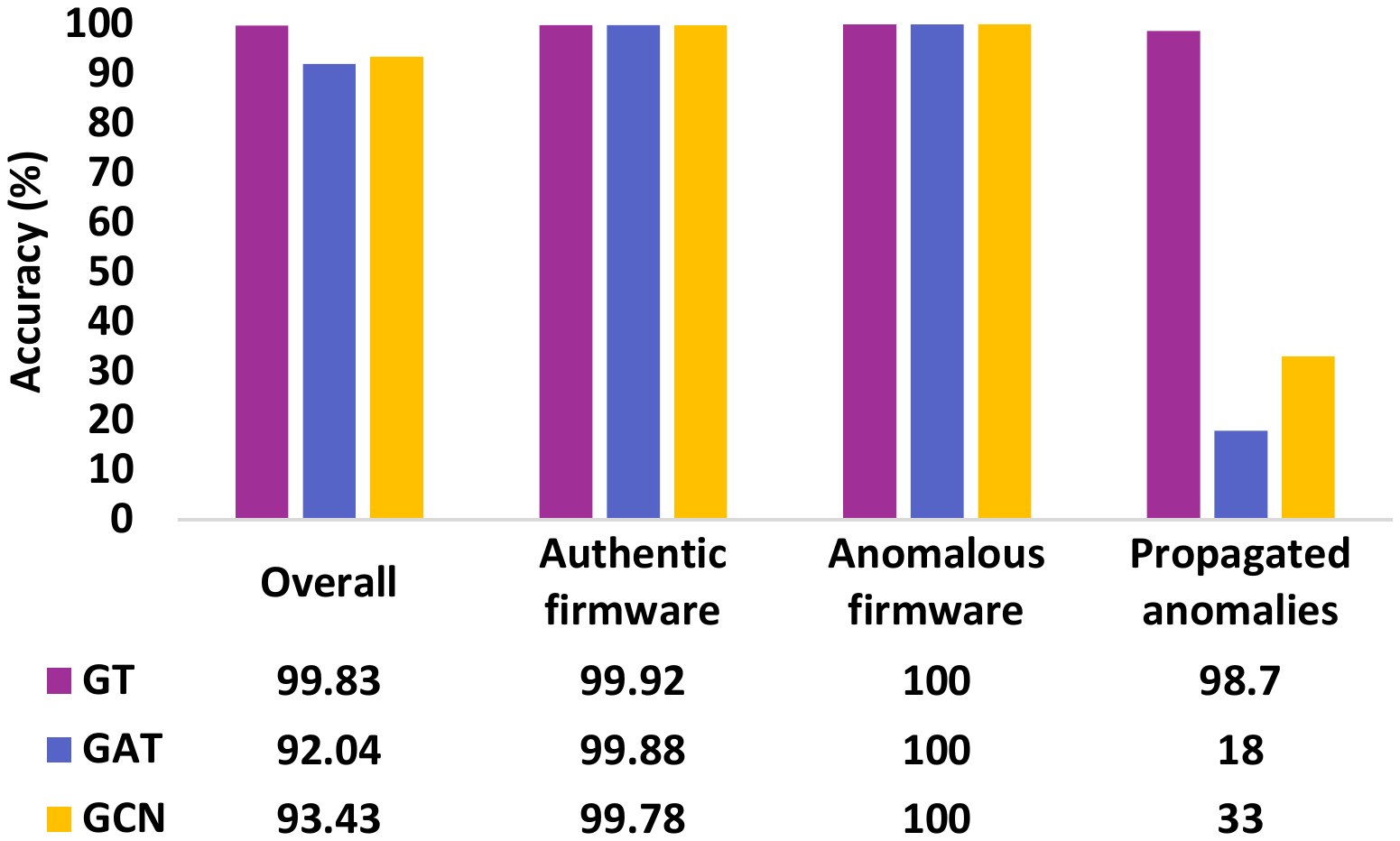}
        \caption{Swarm-1}
        \label{fig:results1}
    \end{subfigure}%
    \begin{subfigure}{\columnwidth}
        \centering
        \includegraphics[trim={0 80 0 80},width=\columnwidth]{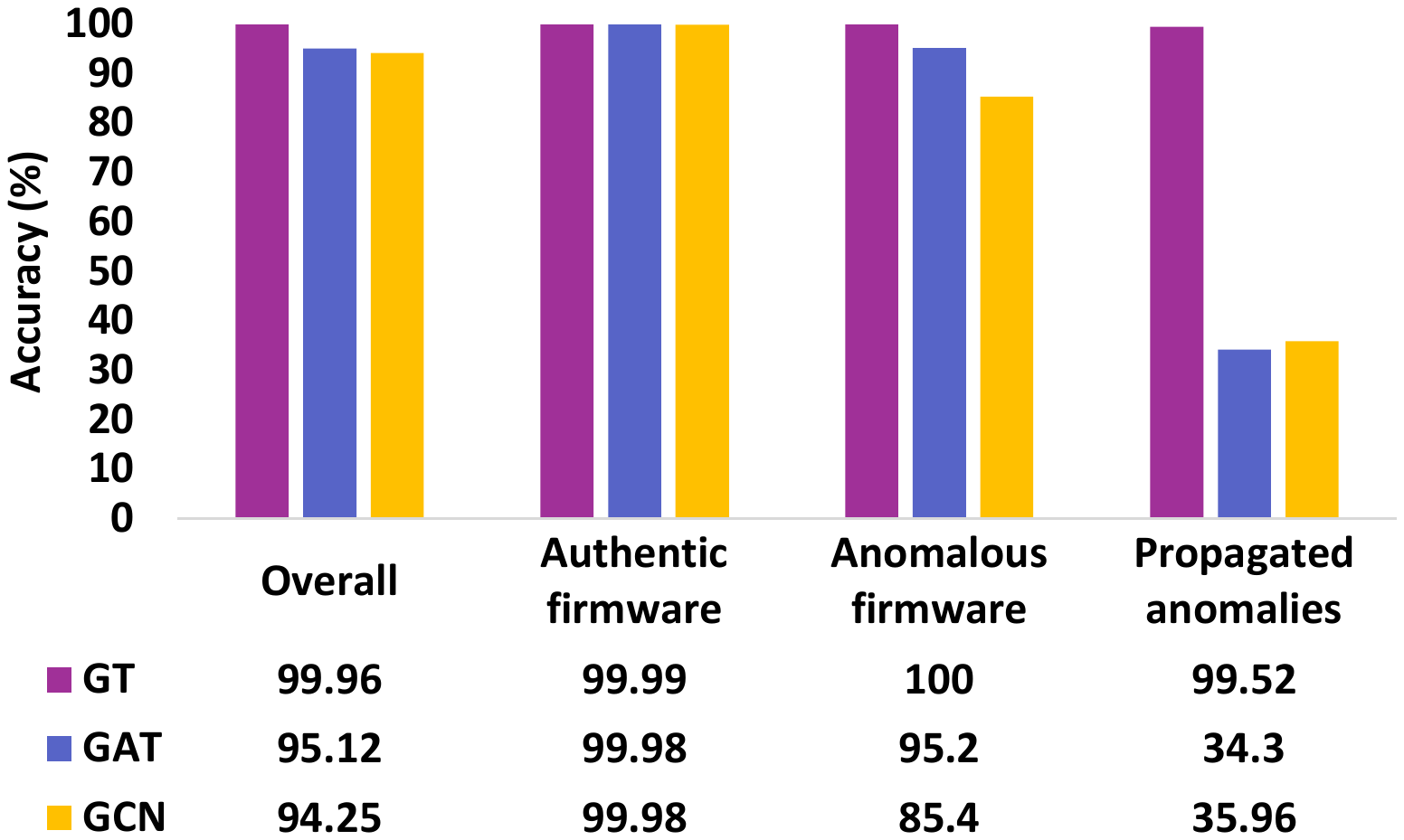}
        \caption{Swarm-2}
        \label{fig:results2}
    \end{subfigure}
    \caption{Performance of the proposed graph models across different detection tasks.}
    \label{fig:results}
\end{figure*}

\begin{figure}[t]
    \centering
    \includegraphics[trim={50 80 50 80},width=0.85\columnwidth]{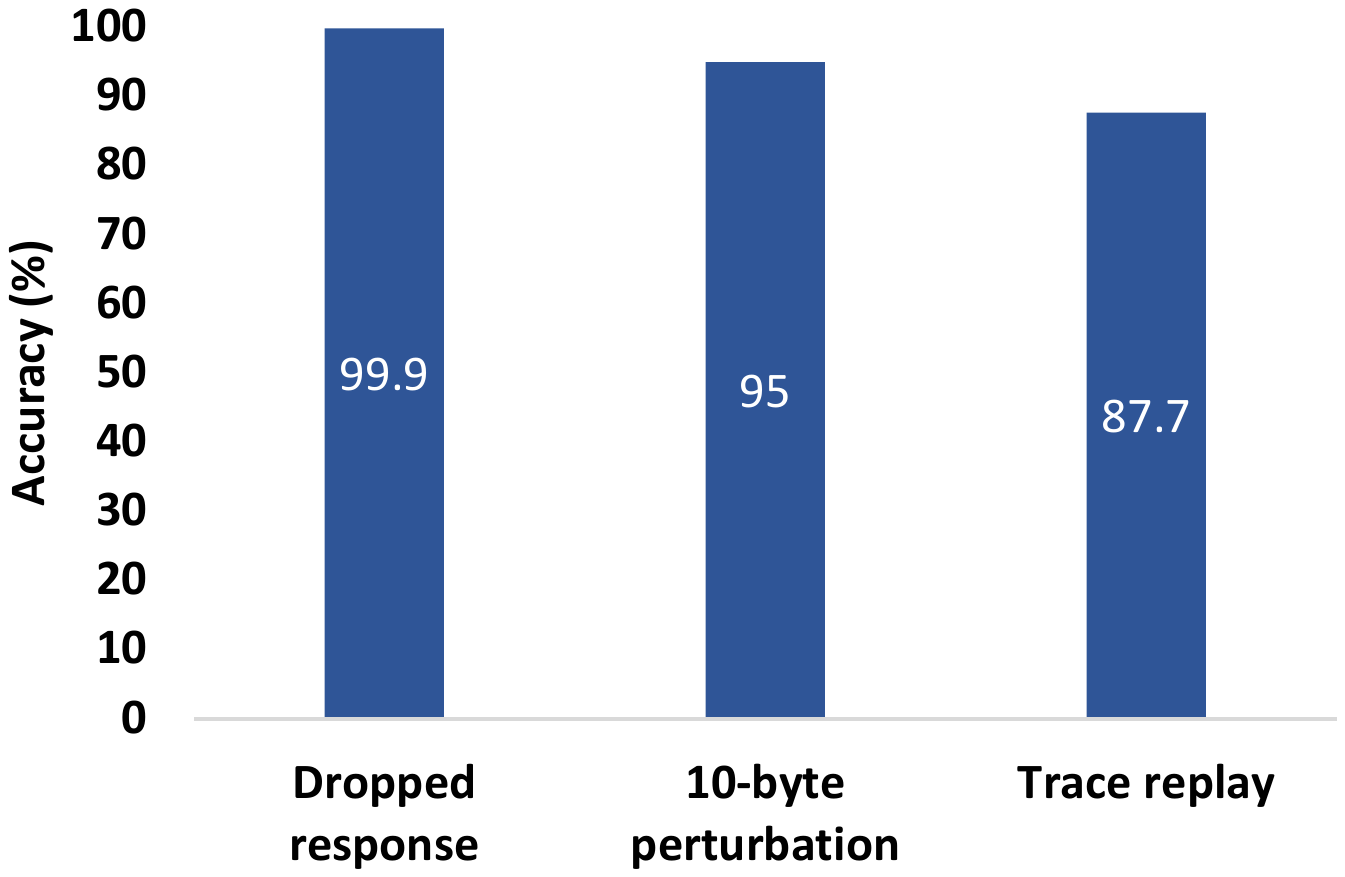}
    \caption{Performance of \emph{Swarm-Net(GT)} on simulated attacks.}
    \label{fig:simulated}
\end{figure}

\subsection{Swarm-2}
Similarly, Table \ref{tab:results2} and Figure \ref{fig:results2} show the results of the three GNN architectures on Swarm-2.

GT achieves an overall accuracy of 99.96\% with a 99.99\% AR on authentic firmware, 100\% DR on anomalous firmware, and 99.52\% DR on propagated anomalies. GAT and GCN have comparable AR on authentic firmware (99.98\% each); however, their performance is lower on anomalous firmware (specifically on $\text{AN}_0$). As mentioned in Section \ref{sec:setup}, the anomalous variant of $\text{N}_0$ in Swarm-2 generates four random bytes compared to six in Swarm-1. This makes the detection of the anomaly more difficult. However, GT has a 100\% DR in this case. Furthermore, given the reduced number of bytes shared between $\text{N}_1$, $\text{N}_2$ and $\text{N}_3$ compared to Swarm-1, detecting downstream anomalies is a greater challenge in this dataset. GAT and GCN have an average 12.4\% and 14.6\% DR on propagated anomalies in $\text{AN}_1$ and $\text{AN}_2$. However, they have 100\% DR in the case of $\text{AN}_4$. Overall, GT outperforms GCN and GAT in detection capabilities in both datasets, making it the best choice.

\begin{table}[t]
\centering
\caption{Comparison with related works.}
\label{tab:comparison}
\resizebox{0.9\columnwidth}{!}{
\begin{tabular}{|cc|c|c|c|} \hline
\multicolumn{2}{|c|}{\textbf{Method}}                              & \multicolumn{1}{c|}{\textbf{$T_o$ (s)}} & \textbf{Model}        & \textbf{Task}                                                   \\ \hline

\multicolumn{2}{|l|}{\emph{Swarm-Net} (CPU)} & $\sim1$ & \multirow{2}{*}{GT} &  \multirow{2}{*}{\begin{tabular}[c]{@{}c@{}}Anomaly\\detection\end{tabular}} \\ \cline{1-3}
\multicolumn{2}{|l|}{\emph{Swarm-Net} (GPU)} & $\sim1$ & & \\\hline

\multicolumn{2}{|l|}{WISE \cite{ammar2018wise}}                    & $3.5$                                                                               & \begin{tabular}[c]{@{}c@{}}Crypto,\\ML\end{tabular}            & \begin{tabular}[c]{@{}c@{}}Checksum,\\clustering\end{tabular}   \\ \hline
\multicolumn{2}{|l|}{FeSA \cite{kuang2022fesa}}                    & $\sim 1$                                                                                 & FL                    & Classification                                                  \\ \hline
\multicolumn{2}{|l|}{Protogerou et. al. \cite{protogerou2021graph}}                   & $> 5$                                                                 & GNN             & \begin{tabular}[c]{@{}c@{}}Anomaly\\detection\end{tabular}      \\ \hline
\multicolumn{2}{|l|}{Rage \cite{chilese2024one}}                   & $t_c+0.15$                                                                              & VGAE                  & \begin{tabular}[c]{@{}c@{}}Anomaly\\detection\end{tabular}      \\ \hline
\multicolumn{2}{|l|}{Aman et. al. \cite{aman2022machine}}          & $\sim 2$                                                                            & MLP                   & Classification                                                  \\ \hline
\multicolumn{2}{|l|}{HAtt \cite{aman2020hatt}}                     & \multicolumn{1}{c|}{$0.126$}                                      & Crypto                & PUF                                                             \\ \hline
\multicolumn{2}{|l|}{SWATT \cite{seshadri2004swatt}}               & $81$                                          & Crypto                 & Checksum   \\ \hline                                                 
\end{tabular}}
\end{table}

\subsection{Simulations}
\label{sec:simulated}

Figure \ref{fig:simulated} shows \emph{Swarm-Net(GT)}'s performance in the simulated scenarios. 

\begin{enumerate}
    \item [1.] \underline{Dropped response ($\text{S}_1$)}: The proposed model has a 99.9\% accuracy in attesting swarms when one node's response is dropped at random. It is, therefore, resistant to single points of failure due to dropped messages.
    \item [2.] \underline{Perturbation ($\text{S}_2$)}: We observe a 95+\% DR for 10 or more bytes of random SRAM perturbation. 
    \item [3.] \underline{Trace replay/out-of-sync states ($\text{S}_3$}): The proposed model achieves an 87.7\% DR in detecting trace replay attacks and out-of-sync network states. Thus, attempts at such threats are likely to be detected.
\end{enumerate}

\subsection{Latency and Memory}
Overhead ($T_o$) may be defined as the summation of the communication overhead ($t_c$), preprocessing ($t_p$), and inference time ($t_i$) following the equation:

\begin{equation}
    T_o = t_c + t_p + t_i
\end{equation}

As shown in Table \ref{tab:comparison}, the overall observed latency in \emph{Swarm-Net} is $\sim 1$ second, most of which is associated with the communication overhead between the verifier, gateway, and swarm. The processing time on the IoT devices (included in $t_c$) is low ($\sim$ milliseconds) since they need not perform complex computation on the SRAM traces, unlike other software-based RA methods \cite{seshadri2004swatt,seshadri2006scuba,khodari2019decentralized}. In addition, the proposed method uses a single SRAM trace for attestation, compared to WISE \cite{protogerou2021graph}, which requires over 5 seconds of network flow information for attestation. Moreover, the communication overhead may change based on the channels used. However, given the data rates ($\sim$ Kbps to Mbps) of existing wireless communication technologies for low-power devices, we expect the overhead to have a similar value. It is worth noting that the communication overhead stays constant for larger swarms since the gateway sends attestation requests in parallel. The evaluation time ($t_p + t_i$) of our best model (GT) is of the order $10^{-5}$ for both CPU- and GPU-based verifiers, which can be attributed to its simple design (compared to the more complex VGAE proposed in RAGE \cite{chilese2024one}). In addition, this simplicity also makes \emph{Swarm-Net} scalable to larger swarm sizes. We also show the setup latency and required storage to give insight into the scalability on the verifier end for a larger number of swarms. The verifier needs around 10 minutes to sample enough traces for training the GNN and incurs $0.036$ seconds (GPU) and $0.563$ seconds (CPU) per epoch (for GT) during the training phase. The proposed model consists of $600,448$ trainable parameters and occupies $\sim2.4$ MB of verifier memory per swarm.


\section{Security Analysis}
\label{sec:security}
We now present a security analysis of the \emph{Swarm-Net} attestation protocol.
\begin{lemma}
\label{l:cons}
    Consistency: The data sections obtained from the same firmware on physical twins behave similarly.
\end{lemma}
The same firmware loaded on two different devices with the same functionality can generate data sections following similar behavior. The results of the proposed attestation algorithm on $P_1$ and $P_2$ from the Swarm-1 dataset support this claim.

\begin{lemma}
\label{l:dist}
    Distinguishability: Only an authorized firmware can generate an authentic SRAM trace using its own SRAM.
\end{lemma}
Changes made to the variables and control dependencies in the firmware create observable differences in the SRAM behavior, as supported by the results on node-level anomalies $\text{AN}_j$ in Section \ref{sec:results}. Thus, an unauthorized firmware cannot generate a valid $T'$ using its own SRAM.

\begin{theorem}
    Mutual Authentication: Completing one protocol run implies that the verifier, gateway, and IoT node have done so with a legitimate counterpart.   
\end{theorem}
\begin{proof}
Since the communication between the verifier and gateway is assumed to be secure, an adversary may only impersonate either (1) the trusted gateway or (2) an IoT node $\text{N}_j$.

In case (1), the adversary must furnish a valid parameter $I_{req,j}$ and subsequently $I_{update,j}$ and $m_{update,j}$ to $\text{N}_j$, which is not possible without the knowledge of the shared secret key, $K_{j}$, and the nonce for that attestation round, $C_j$.

In case (2), the adversary must furnish a valid parameter $I_{resp,j}$ and message $m_{resp,j}$ to the gateway, which is not possible without knowledge of the shared secret key, $K_{j}$, and the updated nonce for the attestation round, $C_j$.
\end{proof}

\begin{theorem}
    Availability: A registered swarm is always available.
\end{theorem}
\begin{proof}
To affect the availability of a registered swarm, an adversary may do one of three things: (1) denial of service on the swarm, (2) replay old attestation requests from the gateway, or (3) drop messages between the gateway and the swarm.

In case (1), an adversary may launch a denial of service on the swarm by sending frequent authentication requests to the swarm, which is not possible without the knowledge of $K_{j}$ and $C_j$ between the gateway and each $\text{N}_j$ in the swarm.

In case (2), an adversary may replay previously valid attestation parameters $I_{req,j}$ to each $\text{N}_j$ in the swarm. However, these parameters will not be valid for subsequent attestation rounds as the nonce $C_j$ is updated by the gateway and the swarm nodes after every attestation round.

Lastly, in case (3), an adversary may drop messages between the gateway and the swarm, causing a state of desynchronization. However, the gateway maintains a list of valid nonces $L_R$ as mentioned in Section \ref{sec:proposed}, which it uses to resynchronize the devices in future attestation rounds. Furthermore, the verifier maintains a list of default traces $T_{def}$, which it uses to replace missing responses during attestation, thereby avoiding single points of failure as shown by the results on dropped responses in Section \ref{sec:simulated}.
\end{proof}

\begin{theorem}
    Attestation: Successful attestation by Algorithm \ref{algo:attest} proves that the provers have authentic firmware.
\end{theorem}
\begin{proof}
An adversary may attempt the following: (1) generate an authentic trace from the SRAM of a device loaded with malicious firmware, and (2) attempt to capture and replay a valid trace.

Case (1) is impossible, as stated in Lemma \ref{l:dist}; malicious firmware cannot generate a valid trace using the microcontroller's SRAM. Furthermore, traces generated from malicious firmware are easily detected by Algorithm \ref{algo:attest} as supported by the results in Section \ref{sec:results}. Finally, attempts to forge an SRAM trace may be unsuccessful based on the results of random perturbation in Section \ref{sec:simulated}.

Further, in case (2), the adversary cannot access the valid traces since the communication between the verifier and the swarm is encrypted. However, if an adversary succeeds, replaying old traces will create out-of-sync network states, which are detectable during attestation based on the results of trace replay attacks in Section \ref{sec:simulated}. 

\end{proof}

\section{Conclusion}
\label{sec:conclusion}
This paper proposed the first SRAM-based swarm attestation approach, \emph{Swarm-Net}, that exploited the graph-like structure of IoT swarms using GNNs. It presented the first datasets on SRAM-based attestation that cover various complicated node-level and inter-node relationships. In addition, a secure protocol was proposed that ensured confidentiality, integrity, mutual authentication, and attestation. \emph{Swarm-Net} achieved a 99.9\% overall accuracy across all types of behaviors ranging from normal firmware to anomalous firmware and propagated anomalies. It was also tested on simulated scenarios, such as dropped responses, trace replay attacks, and SRAM perturbation, which showed resistance to such attacks. Latency evaluation showed an overhead and evaluation latency of the order of $1$ second and $10^{-5}$ seconds, respectively. Lastly, a security analysis highlighted security against impersonation, replay attacks, denial of service, dropped messages, malicious firmware, and propagated anomalies. Future studies can attempt swarm attestation using a smaller number of data section bytes. Furthermore, while this study uses single SRAM responses for attestation, time-series information captured by sequences of SRAM traces may help detect intermittent malicious activity and can be explored by future studies.


\bibliographystyle{IEEEtranN}
{\footnotesize\bibliography{references}}

\begin{thebibliography}{47}
\providecommand{\natexlab}[1]{#1}
\providecommand{\url}[1]{#1}
\csname url@samestyle\endcsname
\providecommand{\newblock}{\relax}
\providecommand{\bibinfo}[2]{#2}
\providecommand{\BIBentrySTDinterwordspacing}{\spaceskip=0pt\relax}
\providecommand{\BIBentryALTinterwordstretchfactor}{4}
\providecommand{\BIBentryALTinterwordspacing}{\spaceskip=\fontdimen2\font plus
\BIBentryALTinterwordstretchfactor\fontdimen3\font minus \fontdimen4\font\relax}
\providecommand{\BIBforeignlanguage}[2]{{%
\expandafter\ifx\csname l@#1\endcsname\relax
\typeout{** WARNING: IEEEtranN.bst: No hyphenation pattern has been}%
\typeout{** loaded for the language `#1'. Using the pattern for}%
\typeout{** the default language instead.}%
\else
\language=\csname l@#1\endcsname
\fi
#2}}
\providecommand{\BIBdecl}{\relax}
\BIBdecl

\bibitem[Chamola et~al.(2020)Chamola, Hassija, Gupta, and Guizani]{chamola2020comprehensive}
V.~Chamola, V.~Hassija, V.~Gupta, and M.~Guizani, ``A comprehensive review of the covid-19 pandemic and the role of iot, drones, ai, blockchain, and 5g in managing its impact,'' \emph{Ieee access}, vol.~8, pp. 90\,225--90\,265, 2020.

\bibitem[Marr()]{forbes}
\BIBentryALTinterwordspacing
B.~Marr. 2024 iot and smart device trends: What you need to know for the future. [Online]. Available: \url{https://www.forbes.com/sites/bernardmarr/2023/10/19/2024-iot-and-smart-device-trends-what-you-need-to-know-for-the-future/}
\BIBentrySTDinterwordspacing

\bibitem[Hassija et~al.(2019)Hassija, Chamola, Saxena, Jain, Goyal, and Sikdar]{hassija2019survey}
V.~Hassija, V.~Chamola, V.~Saxena, D.~Jain, P.~Goyal, and B.~Sikdar, ``A survey on iot security: application areas, security threats, and solution architectures,'' \emph{IEEe Access}, vol.~7, pp. 82\,721--82\,743, 2019.

\bibitem[Ilascu(2019)]{ilascu2019their}
L.~Ilascu, ``When their firmware is vulnerable, its up to you to protect your smart devices,'' \emph{Accessed: May}, vol.~5, 2019.

\bibitem[Ankerg{\aa}rd et~al.(2021)Ankerg{\aa}rd, Dushku, and Dragoni]{ankergaard2021state}
S.~F. J.~J. Ankerg{\aa}rd, E.~Dushku, and N.~Dragoni, ``State-of-the-art software-based remote attestation: Opportunities and open issues for internet of things,'' \emph{Sensors}, vol.~21, no.~5, p. 1598, 2021.

\bibitem[Sfyrakis and Gross(2020)]{sfyrakis2020survey}
I.~Sfyrakis and T.~Gross, ``A survey on hardware approaches for remote attestation in network infrastructures,'' \emph{arXiv preprint arXiv:2005.12453}, 2020.

\bibitem[Johnson et~al.(2021)Johnson, Ghafoor, and Prowell]{johnson2021taxonomy}
W.~A. Johnson, S.~Ghafoor, and S.~Prowell, ``A taxonomy and review of remote attestation schemes in embedded systems,'' \emph{IEEE Access}, vol.~9, pp. 142\,390--142\,410, 2021.

\bibitem[Ambrosin et~al.(2020)Ambrosin, Conti, Lazzeretti, Rabbani, and Ranise]{ambrosin2020collective}
M.~Ambrosin, M.~Conti, R.~Lazzeretti, M.~M. Rabbani, and S.~Ranise, ``Collective remote attestation at the internet of things scale: State-of-the-art and future challenges,'' \emph{IEEE Communications Surveys \& Tutorials}, vol.~22, no.~4, pp. 2447--2461, 2020.

\bibitem[Seshadri et~al.(2004)Seshadri, Perrig, Van~Doorn, and Khosla]{seshadri2004swatt}
A.~Seshadri, A.~Perrig, L.~Van~Doorn, and P.~Khosla, ``Swatt: Software-based attestation for embedded devices,'' in \emph{IEEE Symposium on Security and Privacy, 2004. Proceedings. 2004}.\hskip 1em plus 0.5em minus 0.4em\relax IEEE, 2004, pp. 272--282.

\bibitem[Seshadri et~al.(2006)Seshadri, Luk, Perrig, Van~Doorn, and Khosla]{seshadri2006scuba}
A.~Seshadri, M.~Luk, A.~Perrig, L.~Van~Doorn, and P.~Khosla, ``Scuba: Secure code update by attestation in sensor networks,'' in \emph{Proceedings of the 5th ACM workshop on Wireless security}, 2006, pp. 85--94.

\bibitem[Seshadri et~al.(2008)Seshadri, Luk, and Perrig]{seshadri2008sake}
A.~Seshadri, M.~Luk, and A.~Perrig, ``Sake: Software attestation for key establishment in sensor networks,'' in \emph{Distributed Computing in Sensor Systems: 4th IEEE International Conference, DCOSS 2008 Santorini Island, Greece, June 11-14, 2008 Proceedings 4}.\hskip 1em plus 0.5em minus 0.4em\relax Springer, 2008, pp. 372--385.

\bibitem[Agrawal et~al.(2015)Agrawal, Das, Mathuria, and Srivastava]{agrawal2015program}
S.~Agrawal, M.~L. Das, A.~Mathuria, and S.~Srivastava, ``Program integrity verification for detecting node capture attack in wireless sensor network,'' in \emph{Information Systems Security: 11th International Conference, ICISS 2015, Kolkata, India, December 16-20, 2015. Proceedings 11}.\hskip 1em plus 0.5em minus 0.4em\relax Springer, 2015, pp. 419--440.

\bibitem[Brasser et~al.(2015)Brasser, El~Mahjoub, Sadeghi, Wachsmann, and Koeberl]{brasser2015tytan}
F.~Brasser, B.~El~Mahjoub, A.-R. Sadeghi, C.~Wachsmann, and P.~Koeberl, ``Tytan: Tiny trust anchor for tiny devices,'' in \emph{Proceedings of the 52nd annual design automation conference}, 2015, pp. 1--6.

\bibitem[Kibret(2023)]{kibret2023property}
S.~W. Kibret, ``Property-based attestation in device swarms: a machine learning approach,'' \emph{Machine Learning for Cyber Security}, p.~71, 2023.

\bibitem[Protogerou et~al.(2021)Protogerou, Papadopoulos, Drosou, Tzovaras, and Refanidis]{protogerou2021graph}
A.~Protogerou, S.~Papadopoulos, A.~Drosou, D.~Tzovaras, and I.~Refanidis, ``A graph neural network method for distributed anomaly detection in iot,'' \emph{Evolving Systems}, vol.~12, no.~1, pp. 19--36, 2021.

\bibitem[Aman et~al.(2022)Aman, Basheer, Wong, Xu, Lim, and Sikdar]{aman2022machine}
M.~N. Aman, H.~Basheer, J.~W. Wong, J.~Xu, H.~W. Lim, and B.~Sikdar, ``Machine-learning-based attestation for the internet of things using memory traces,'' \emph{IEEE Internet of Things Journal}, vol.~9, no.~20, pp. 20\,431--20\,443, 2022.

\bibitem[Krau{\ss} et~al.(2007)Krau{\ss}, Stumpf, and Eckert]{krauss2007detecting}
C.~Krau{\ss}, F.~Stumpf, and C.~Eckert, ``Detecting node compromise in hybrid wireless sensor networks using attestation techniques,'' in \emph{Security and Privacy in Ad-hoc and Sensor Networks: 4th European Workshop, ESAS 2007, Cambridge, UK, July 2-3, 2007. Proceedings 4}.\hskip 1em plus 0.5em minus 0.4em\relax Springer, 2007, pp. 203--217.

\bibitem[Tan et~al.(2011)Tan, Hu, and Jha]{tan2011tpm}
H.~Tan, W.~Hu, and S.~Jha, ``A tpm-enabled remote attestation protocol (trap) in wireless sensor networks,'' in \emph{Proceedings of the 6th ACM workshop on Performance monitoring and measurement of heterogeneous wireless and wired networks}, 2011, pp. 9--16.

\bibitem[Eldefrawy et~al.(2012)Eldefrawy, Tsudik, Francillon, and Perito]{eldefrawy2012smart}
K.~Eldefrawy, G.~Tsudik, A.~Francillon, and D.~Perito, ``Smart: secure and minimal architecture for (establishing dynamic) root of trust.'' in \emph{Ndss}, vol.~12, 2012, pp. 1--15.

\bibitem[Koeberl et~al.(2014)Koeberl, Schulz, Sadeghi, and Varadharajan]{koeberl2014trustlite}
P.~Koeberl, S.~Schulz, A.-R. Sadeghi, and V.~Varadharajan, ``Trustlite: A security architecture for tiny embedded devices,'' in \emph{Proceedings of the Ninth European Conference on Computer Systems}, 2014, pp. 1--14.

\bibitem[Aman et~al.(2020)Aman, Basheer, Dash, Wong, Xu, Lim, and Sikdar]{aman2020hatt}
M.~N. Aman, M.~H. Basheer, S.~Dash, J.~W. Wong, J.~Xu, H.~W. Lim, and B.~Sikdar, ``Hatt: Hybrid remote attestation for the internet of things with high availability,'' \emph{IEEE Internet of Things Journal}, vol.~7, no.~8, pp. 7220--7233, 2020.

\bibitem[Khodari et~al.(2019)Khodari, Rawat, Asplund, and Gurtov]{khodari2019decentralized}
M.~Khodari, A.~Rawat, M.~Asplund, and A.~Gurtov, ``Decentralized firmware attestation for in-vehicle networks,'' in \emph{Proceedings of the 5th on Cyber-Physical System Security Workshop}, 2019, pp. 47--56.

\bibitem[Asokan et~al.(2015)Asokan, Brasser, Ibrahim, Sadeghi, Schunter, Tsudik, and Wachsmann]{asokan2015seda}
N.~Asokan, F.~Brasser, A.~Ibrahim, A.-R. Sadeghi, M.~Schunter, G.~Tsudik, and C.~Wachsmann, ``Seda: Scalable embedded device attestation,'' in \emph{Proceedings of the 22nd ACM SIGSAC conference on computer and communications security}, 2015, pp. 964--975.

\bibitem[Carpent et~al.(2017)Carpent, ElDefrawy, Rattanavipanon, and Tsudik]{carpent2017lightweight}
X.~Carpent, K.~ElDefrawy, N.~Rattanavipanon, and G.~Tsudik, ``Lightweight swarm attestation: A tale of two lisa-s,'' in \emph{Proceedings of the 2017 ACM on Asia Conference on Computer and Communications Security}, 2017, pp. 86--100.

\bibitem[Visintin et~al.(2019)Visintin, Toffalini, Conti, and Zhou]{visintin2019safe}
A.~Visintin, F.~Toffalini, M.~Conti, and J.~Zhou, ``Safe\^{} d: Self-attestation for networks of heterogeneous embedded devices,'' \emph{arXiv preprint arXiv:1909.08168}, 2019.

\bibitem[Kuang et~al.(2019)Kuang, Fu, Yu, Yang, Su, and Zhang]{kuang2019esdra}
B.~Kuang, A.~Fu, S.~Yu, G.~Yang, M.~Su, and Y.~Zhang, ``Esdra: An efficient and secure distributed remote attestation scheme for iot swarms,'' \emph{IEEE Internet of Things Journal}, vol.~6, no.~5, pp. 8372--8383, 2019.

\bibitem[Dushku et~al.(2020)Dushku, Rabbani, Conti, Mancini, and Ranise]{dushku2020sara}
E.~Dushku, M.~M. Rabbani, M.~Conti, L.~V. Mancini, and S.~Ranise, ``Sara: Secure asynchronous remote attestation for iot systems,'' \emph{IEEE Transactions on Information Forensics and Security}, vol.~15, pp. 3123--3136, 2020.

\bibitem[Ibrahim et~al.(2019)Ibrahim, Sadeghi, and Tsudik]{ibrahim2019healed}
A.~Ibrahim, A.-R. Sadeghi, and G.~Tsudik, ``Healed: Healing \& attestation for low-end embedded devices,'' in \emph{Financial Cryptography and Data Security: 23rd International Conference, FC 2019, Frigate Bay, St. Kitts and Nevis, February 18--22, 2019, Revised Selected Papers 23}.\hskip 1em plus 0.5em minus 0.4em\relax Springer, 2019, pp. 627--645.

\bibitem[Ammar et~al.(2018)Ammar, Washha, and Crispo]{ammar2018wise}
M.~Ammar, M.~Washha, and B.~Crispo, ``Wise: Lightweight intelligent swarm attestation scheme for iot (the verifier’s perspective),'' in \emph{2018 14th International Conference on Wireless and Mobile Computing, Networking and Communications (WiMob)}.\hskip 1em plus 0.5em minus 0.4em\relax IEEE, 2018, pp. 1--8.

\bibitem[Kuang et~al.(2022)Kuang, Fu, Gao, Zhang, Zhou, and Deng]{kuang2022fesa}
B.~Kuang, A.~Fu, Y.~Gao, Y.~Zhang, J.~Zhou, and R.~H. Deng, ``Fesa: Automatic federated swarm attestation on dynamic large-scale iot devices,'' \emph{IEEE Transactions on Dependable and Secure Computing}, vol.~20, no.~4, pp. 2954--2969, 2022.

\bibitem[Chilese et~al.(2024)Chilese, Mitev, Orenbach, Thorburn, Atamli, and Sadeghi]{chilese2024one}
M.~Chilese, R.~Mitev, M.~Orenbach, R.~Thorburn, A.~Atamli, and A.-R. Sadeghi, ``One for all and all for one: Gnn-based control-flow attestation for embedded devices,'' \emph{arXiv preprint arXiv:2403.07465}, 2024.

\bibitem[Jakobsson and Johansson(2010)]{jakobsson2010retroactive}
M.~Jakobsson and K.-A. Johansson, ``Retroactive detection of malware with applications to mobile platforms,'' in \emph{5th USENIX Workshop on Hot Topics in Security (HotSec 10)}, 2010.

\bibitem[Li et~al.(2010)Li, McCune, and Perrig]{li2010sbap}
Y.~Li, J.~M. McCune, and A.~Perrig, ``Sbap: Software-based attestation for peripherals,'' in \emph{Trust and Trustworthy Computing: Third International Conference, TRUST 2010, Berlin, Germany, June 21-23, 2010. Proceedings 3}.\hskip 1em plus 0.5em minus 0.4em\relax Springer, 2010, pp. 16--29.

\bibitem[Chen et~al.(2017)Chen, Dong, Bai, Jauhar, and Cheng]{chen2017secure}
B.~Chen, X.~Dong, G.~Bai, S.~Jauhar, and Y.~Cheng, ``Secure and efficient software-based attestation for industrial control devices with arm processors,'' in \emph{Proceedings of the 33rd Annual Computer Security Applications Conference}, 2017, pp. 425--436.

\bibitem[Aman and Sikdar(2018)]{aman2018att}
M.~N. Aman and B.~Sikdar, ``Att-auth: A hybrid protocol for industrial iot attestation with authentication,'' \emph{IEEE Internet of Things Journal}, vol.~5, no.~6, pp. 5119--5131, 2018.

\bibitem[Hristozov et~al.(2018)Hristozov, Heyszl, Wagner, and Sigl]{hristozov2018practical}
S.~Hristozov, J.~Heyszl, S.~Wagner, and G.~Sigl, ``Practical runtime attestation for tiny iot devices,'' in \emph{NDSS Workshop on Decentralized IoT Security and Standards (DISS)}, vol.~18, 2018.

\bibitem[Koshy and Pandey(2005)]{koshy2005remote}
J.~Koshy and R.~Pandey, ``Remote incremental linking for energy-efficient reprogramming of sensor networks,'' in \emph{Proceeedings of the Second European Workshop on Wireless Sensor Networks, 2005.}\hskip 1em plus 0.5em minus 0.4em\relax IEEE, 2005, pp. 354--365.

\bibitem[Kohli et~al.(2024{\natexlab{a}})Kohli, Aman, and Sikdar]{kohli2024intelligent}
V.~Kohli, M.~N. Aman, and B.~Sikdar, ``An intelligent fingerprinting technique for low-power embedded iot devices,'' \emph{IEEE Transactions on Artificial Intelligence}, 2024.

\bibitem[Wu et~al.(2020)Wu, Pan, Chen, Long, Zhang, and Philip]{wu2020comprehensive}
Z.~Wu, S.~Pan, F.~Chen, G.~Long, C.~Zhang, and S.~Y. Philip, ``A comprehensive survey on graph neural networks,'' \emph{IEEE transactions on neural networks and learning systems}, vol.~32, no.~1, pp. 4--24, 2020.

\bibitem[Kipf and Welling(2016)]{kipf2016semi}
T.~N. Kipf and M.~Welling, ``Semi-supervised classification with graph convolutional networks,'' \emph{arXiv preprint arXiv:1609.02907}, 2016.

\bibitem[Veli{\v{c}}kovi{\'c} et~al.(2017)Veli{\v{c}}kovi{\'c}, Cucurull, Casanova, Romero, Lio, and Bengio]{velivckovic2017graph}
P.~Veli{\v{c}}kovi{\'c}, G.~Cucurull, A.~Casanova, A.~Romero, P.~Lio, and Y.~Bengio, ``Graph attention networks,'' \emph{arXiv preprint arXiv:1710.10903}, 2017.

\bibitem[Shi et~al.(2020)Shi, Huang, Feng, Zhong, Wang, and Sun]{shi2020masked}
Y.~Shi, Z.~Huang, S.~Feng, H.~Zhong, W.~Wang, and Y.~Sun, ``Masked label prediction: Unified message passing model for semi-supervised classification,'' \emph{arXiv preprint arXiv:2009.03509}, 2020.

\bibitem[Vaswani et~al.(2017)Vaswani, Shazeer, Parmar, Uszkoreit, Jones, Gomez, Kaiser, and Polosukhin]{vaswani2017attention}
A.~Vaswani, N.~Shazeer, N.~Parmar, J.~Uszkoreit, L.~Jones, A.~N. Gomez, {\L}.~Kaiser, and I.~Polosukhin, ``Attention is all you need,'' \emph{Advances in neural information processing systems}, vol.~30, 2017.

\bibitem[Samiullah et~al.(2023)Samiullah, Gan, Akleylek, and Aun]{samiullah2023group}
F.~Samiullah, M.~L. Gan, S.~Akleylek, and Y.~Aun, ``Group key management in internet of things: A systematic literature review,'' \emph{IEEE Access}, 2023.

\bibitem[Gilbert and Handschuh(2003)]{gilbert2003security}
H.~Gilbert and H.~Handschuh, ``Security analysis of sha-256 and sisters,'' in \emph{International workshop on selected areas in cryptography}.\hskip 1em plus 0.5em minus 0.4em\relax Springer, 2003, pp. 175--193.

\bibitem[Kohli et~al.(2024{\natexlab{b}})Kohli, Kohli, Naveed~Aman, and Sikdar]{gmee-vj41-24}
\BIBentryALTinterwordspacing
V.~Kohli, B.~Kohli, M.~Naveed~Aman, and B.~Sikdar, ``Iot device swarm sram dataset for firmware attestation,'' 2024. [Online]. Available: \url{https://dx.doi.org/10.21227/gmee-vj41}
\BIBentrySTDinterwordspacing

\bibitem[Kingma and Ba(2014)]{kingma2014adam}
D.~P. Kingma and J.~Ba, ``Adam: A method for stochastic optimization,'' \emph{arXiv preprint arXiv:1412.6980}, 2014.

\end{thebibliography}

\end{document}